\documentclass[12pt]{amsart}
\usepackage{amsmath}
\usepackage{amssymb}
\usepackage{amscd}
\usepackage{graphicx}
\usepackage[dvips]{epsfig}
\usepackage{ulem}
\newtheorem{thm}{Theorem}[section]
\newtheorem{lem}[thm]{Lemma}

\newtheorem{defn}[thm]{Definition}

\title{Generalized uncertainty principles}

\author{Ronny Machluf}
\address{Department Of Mathematics, The Weizmann
Institute of Science, Rehovot, Israel}
\email{machlufr@weizmann.ac.il} 
\date{9 July 2008}
\begin{document}
\begin{abstract}
The phenomenon in the essence of classical uncertainty principles
is well known since the thirties of the last century. We introduce
a new phenomenon which is in the essence of a new notion that we
introduce: "Generalized Uncertainty Principles". We show the
relation between classical uncertainty principles and generalized
uncertainty principles. We generalized "Landau-Pollak-Slepian"
uncertainty principle. Our generalization relates the following
two quantities and two scaling parameters: 1) The weighted time
spreading $\int_{-\infty}^\infty |f(x)|^2w_1(x)dx$, ($w_1(x)$ is a
non-negative function). 2) The weighted frequency spreading
$\int_{-\infty}^\infty |\hat{f}(\omega)|^2w_2(\omega)d\omega$. 3)
The time weight scale $a$, ${w_1}_a(x)=w_1(xa^{-1})$ and 4) The
frequency weight scale $b$, ${w_2}_b(\omega)=w_2(\omega b^{-1})$.
"Generalized Uncertainty Principle" is an inequality that
summarizes the constraints on the relations between the two
spreading quantities and two scaling parameters. For any two
reasonable weights $w_1(x)$ and $w_2(\omega)$, we introduced a
three dimensional set in $R^3$ that is in the essence of many
uncertainty principles. The set is called "possibility body". We
showed that classical uncertainty principles (such as the
Heiseneberg-Pauli-Weyl uncertainty principle) stem from lower
bounds for different functions defined on the possibility body. We
investigated qualitative properties of general uncertainty
principles and possibility bodies. Using this approach we derived
new (quantitative) uncertainty principles for
Landau-Pollak-Slepian weights. We found the general uncertainty
principles related to homogeneous weights, $w_1(x)=w_2(x)=x^k$,
$k\in N$, up to a constant.
\end{abstract}
\maketitle

\section{Notations}

 \begin{itemize}

\item[$\circ$]

We denote by $L^2$ the space of functions such that:
$\int_{-\infty}^\infty |f(x)|^2 dx<\infty$ with the inner product
defined by: $\langle f, g \rangle = \int_{-\infty}^\infty
f(x)\overline{g(x)}dx$ and norm defined by: $||f||=\langle f, f
\rangle^\frac{1}{2}$

\item[$\circ$]

Fourier transform: $\hat f(w)=\int_{R^d} f(x)e^{-2\pi i x \cdot
\omega}dx$

\item[$\circ$]
Translation operators: $T_a f(x)=f(x-a)$

\item[$\circ$]
Modulation operators:$M_b f(x)=\exp(ibx)f(x)$

\item[$\circ$]
Scaling operators:$S_a f(x)=f(xa^{-1})$

\item[$\circ$]

Time limiting operators $D_h:L^2\rightarrow L^2$: $\widehat{D_h
f}(\omega)=\int_{-h}^{h} f(x)e^{-2\pi i x \cdot \omega}dx$

\item[$\circ$]
Band limiting operators $B_m:L^2\rightarrow L^2$: $B_m
f(x)=\int_{-m}^{m} \hat f(\omega)e^{-2\pi i x \cdot
\omega}d\omega$

\item[$\circ$]
The function $\cos^{-1}(x):[-1,1]\rightarrow [0,\pi]$ is the
inverse of $\cos(x)$:

$\cos^{-1}(\cos (x))=x$, $\forall x\in [0,\pi]$

\item[$\circ$]

The function $\sin^{-1}(x):[-1,1]\rightarrow
[-\frac{\pi}{2},\frac{\pi}{2}]$ is the inverse of $\sin(x)$:
$\sin^{-1}(\sin (x))=x$, $\forall x\in
[-\frac{\pi}{2},\frac{\pi}{2}]$

\end{itemize}

\section{Introduction}
The meta-principle that a signal can not be localized both at time
and frequency is reflected as inequalities involving a function,
say $f$, and its Fourier transform, $\hat{f}$. We will call that
kind of inequalities $\textit{classical uncertainty}$
$\textit{principles}$. A good survey for the subject by Gerald B.
Folland and Alladi Sitaram is \cite{FollandGB_SitaramA}.
Uncertainty principles uses the notion of concentration (e.g.
Landau-Pollak-Slepian (LPS), see below) or the notion of spreading
(e.g. Heisenberg-Pauli-Weyl (HPW), see below). By "translating"
Landau and Pollak (LP) result from "concentration language" to
"spreading language" on one side and adding two parameters to HPW
result on the other side, we show that those results are special
cases (up to minor changes in the LP case) of what we call
$\textit{Generalized Uncertainty Principles}$. Our approach
explains the qualitative behavior which is related to LP result
and HPW result. The important quantitative results in LP result
and HPW result can not be achieved using our approach. recall the
(classical) Heisenberg-Pauli-Weyl uncertainty principle:

\vspace{3 mm}
\begin{thm}[Heisenberg-Pauli-Weyl Uncertainty Principle]
\label{HPWtheorem}

If $f\in L^2(R)$ and $a,b\in R$ are arbitrary, then

\begin{equation}
(\int_{-\infty}^\infty
(x-a)^2|f(x)|^2dx)^\frac{1}{2}(\int_{-\infty}^\infty
(\omega-b)^2|\hat f(\omega)|^2d\omega)^\frac{1}{2}\geq
\frac{1}{4\pi}||f||_2^2
\end{equation}

Equality holds if and only if $f$ is a multiple of $\exp^{2\pi i
b(x-a)}\exp^{-\pi (x-a)^2/c}$ for some $a,b\in R$ and $c>0$

$\Box$

\end{thm}

\vspace{3 mm}

HPW uncertainty principle was the first uncertainty principle that
appeared (\cite{HeisenbergW, KennardEH} \cite[Appendix 1]{WeylH}).
Different variations of HPW uncertainty principle have been
published since then \cite[section 3]{FollandGB_SitaramA}. LP
\cite{LandauHJ_PollakHO1} defined for every $f\in L^2$ its
$\textit{time concentration}$ by:

\begin{equation}
\label{Time_concentration}
 \alpha_{T}(f)=\frac{\big(\int\limits_{-T}^{T}
|f(x)|^2dx\big)^\frac{1}{2}}{(\int\limits_{-\infty}^\infty
|f(x)|^2dx)^\frac{1}{2}}
\end{equation}

and its $\textit{frequency concentration}$ by:

\begin{equation}
\label{Frequency_concentration}
\beta_\Omega(f)=\frac{(\int\limits_{-\Omega}^\Omega |\hat
f(\omega)|^2d\omega)^\frac{1}{2}}{(\int\limits_{-\infty}^\infty
|\hat f(\omega)|^2d\omega)^\frac{1}{2}}
\end{equation}

and found explicitly the following set of points in $R^2$:

$$M_{\Omega,T}=\bigcup_{f\in L^2}(\alpha_{T}(f),\beta_\Omega(f))$$

We define the complement of $M_{\Omega,T}$ in

$D_{LP}=[0,1]\times[0,1]\setminus\{(0,1),(1,0),(1,1)\}$:

$$M'_{\Omega, T}=D_{LP}\setminus M_{\Omega, T}$$

We will call the pair $(M_{\Omega,T},M'_{\Omega,T})$ a
"possibility map"; where $M_{\Omega, T}$ is the "possible area"
and $M'_{\Omega, T}$ is the "impossible area" (note the dependence
on $\Omega$ and $T$). Of course, for individual $f\in L^2$, the
point $(\alpha_{T}(f),\beta_\Omega(f))$ depends on $T$ and
$\Omega$. LP noticed that the "possibility map" depends on the
product $c=\Omega T$ only and not on $\Omega$ and $T$ separately
($M_c=M_{\Omega,T}$, In \cite{LandauHJ_PollakHO1} the integral
bounds in the numerator of \eqref{Time_concentration} are from
$-\frac{T}{2}$ to $\frac{T}{2}$ and therefore the notation in
their papers is $c=\Omega \frac{T}{2}$). LP
\cite{LandauHJ_PollakHO1} proved an uncertainty principle of the
form:

\begin{equation}
\label{uncertainty0} c\geq\phi(\alpha,\beta)
\end{equation}

$$\phi:D_{LP}\rightarrow[0,\infty)$$

(more details is section \ref{sec3}). We will call inequalities of
the type \eqref{uncertainty0} and its generalizations that we will
introduce below "Generalized uncertainty principles". We will call
the pair $(M_c,M_c')$ "the possibility map of level c", where
$M_c$ is the "possible area of level c" and $M_c'$ is the
"impossible area of level c". We will call the set

$$PB=\{ (\alpha,\beta,c) | (\alpha,\beta) \in M_c, c\in
(0,\infty)\}$$

"a possibility body". In section \ref{sec3} we state LP result in
the original way and in a way that can be generalized. We show how
uncertainty principles (inequalities) are derived by bounding
functions which are defined on the "possible area of level c"
(where c is a parameter in the inequality). We also show that in
one natural coordinate system the "possible area of level c"  is
convex and in another natural coordinate system the "possible area
of level c" is non-convex.

In section~\ref{sec4} we show that for some general weights we
have the same qualitative behavior regarding uncertainty, and
using HPW uncertainty principle we derive the possibility map and
possibility body for the HPW weight ($x^2$). In section \ref{sec5}
we discuss the question of convexity of the possibility body and
the question of the right coordinate system for describing the
"General Uncertainty Principle" phenomenon. We find the general
uncertainty principles for homogenous weights $w_1(x)=w_2(x)=x^k$,
$k\in N$. \textit{Heisenberg-Pauli-Weyl general uncertainty
principle is a special case that corresponds to $k=2$.}

\section{Slepian-Pollak-Landau uncertainty principle}
\label{sec3}

 In a series of papers by LPS \cite{LandauHJ_PollakHO1,LandauHJ_PollakHO2,SlepianD_PollakHO} the
following integral equation was investigated:

\begin{equation}\label{IntegralEquation}
\lambda
f(t)=\frac{1}{\pi}\int^{T}_{-T}f(s)\frac{sin\Omega(t-s)}{t-s}ds
\end{equation}

on $L^2[-T,T]$

\vspace{7 mm}

They have showed that the eigenvalues of \eqref{IntegralEquation}
are distinct, positive and depend on the product $c=\Omega T$. (In
\cite{LandauHJ_PollakHO1} the integral bounds of
\eqref{IntegralEquation} are from $-\frac{T}{2}$ to $\frac{T}{2}$
and therefore the notation in their papers is $2c=\Omega T$).

LP \cite{LandauHJ_PollakHO1} proved the following uncertainty
principle:

\begin {thm}[LP Theorem]
\label{LPtheorem}
 There is a function $f$ such that
$||f||=1$, $\alpha_T(f)=\alpha$ and $\beta_\Omega(f)=\beta$, under
the following conditions, and only under the following conditions:

1) If $\alpha=0$                     and $0\leq \beta<1$

2) If $0<\alpha<\sqrt{\lambda_0}$    and $0\leq \beta \leq 1$

3) If $\sqrt{\lambda_0}\leq\alpha<1$   and
$\cos^{-1}\alpha+\cos^{-1}\beta\geq \cos^{-1}\sqrt{\lambda_0}$

4) If $\alpha=1$                       and $0<\beta\leq
\sqrt{\lambda_0}$

where $\lambda_0$ is the largest eigenvalue of
\eqref{IntegralEquation}

$\Box$
\end{thm}

\vspace{7 mm}

The complexity of computing the biggest eigenvalue, $\lambda_0$,
of \eqref{IntegralEquation} for fixed $c$ grows rapidly with c.
Algorithms for computing $\lambda_0$ for fixed $c$ can be found at
\cite{BouwkampCJ,KhareK_GeorgeN},\cite[and the references
therein]{KarouiA_MoumniT}. There are no error estimates and no
complexity analysis at the literature. W.H.J.Fuchs \cite{FuchsWHJ}
proved the following asymptotic formula:

\begin {thm}[Fuchs Theorem]
\label{Fuchs_theorem}
 Let $\lambda_0\geq \lambda_1\geq \lambda_2
 \geq \lambda_3 ...$ be the eigenvalues of the integral equation
 \eqref{IntegralEquation}. Then
\begin{equation}
\label{asymptotic_formula}
 1-\lambda_n\sim 4\pi^{1/2}8^n(n!)^{-1}c^{n+1/2}e^{2c}
\end{equation}
as $c\rightarrow \infty$.

$\Box$
\end{thm}

We used H.Xiao, V.Rokhlin and N.Yarvin (XRY)
\cite{XiaoH_RokhlinV_YarvinN} algorithm for computing $\lambda_0$
as a function of $c\in[0,6]$ and compered it to the asymptotic
formula \eqref{asymptotic_formula} (see Figure
\ref{lambda_0_as_a_function_of_c_fig}). We used $k=300$ for the
complexity parameter in XRY \cite{XiaoH_RokhlinV_YarvinN}
algorithm (see equation (54) and section 4 therein) . The
asymptotic formula is a good approximation for $\lambda_0$
starting from small $c's$. The relative difference (i.e.
(Numerical $\lambda_0$ - asymptotic $\lambda_0$)/(1- asymptotic
$\lambda_0$)) is decreasing up to $c=5$ and then it starts to
increase. We believe that with higher complexity resources then us
(We used PC) one should use numerical algorithms for $c\leq 10$.
We believe that for $c>10$ the asymptotic formula is good enough.

 Possibility maps for
different values of $\lambda_0$ are plotted in Figure
\ref{Concentration_alpha2_possibility_maps_fig}. Below we will
introduce possibility maps in term of spreading. This is an
important difference between our approach to LPS approach.

\begin{figure}[!h]
  \centerline{
    \mbox{\includegraphics[scale=0.4]{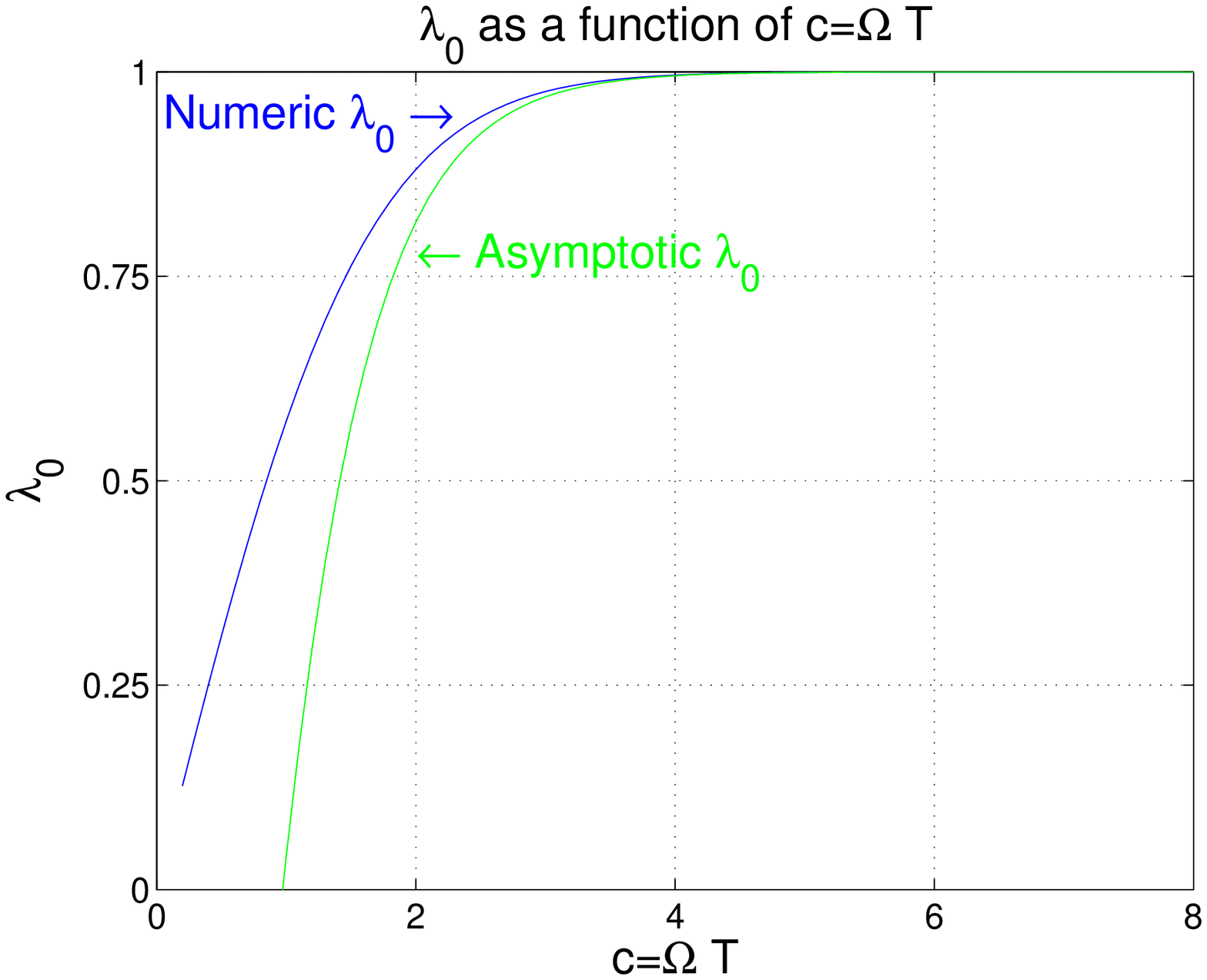}}
    \mbox{\includegraphics[scale=0.4]{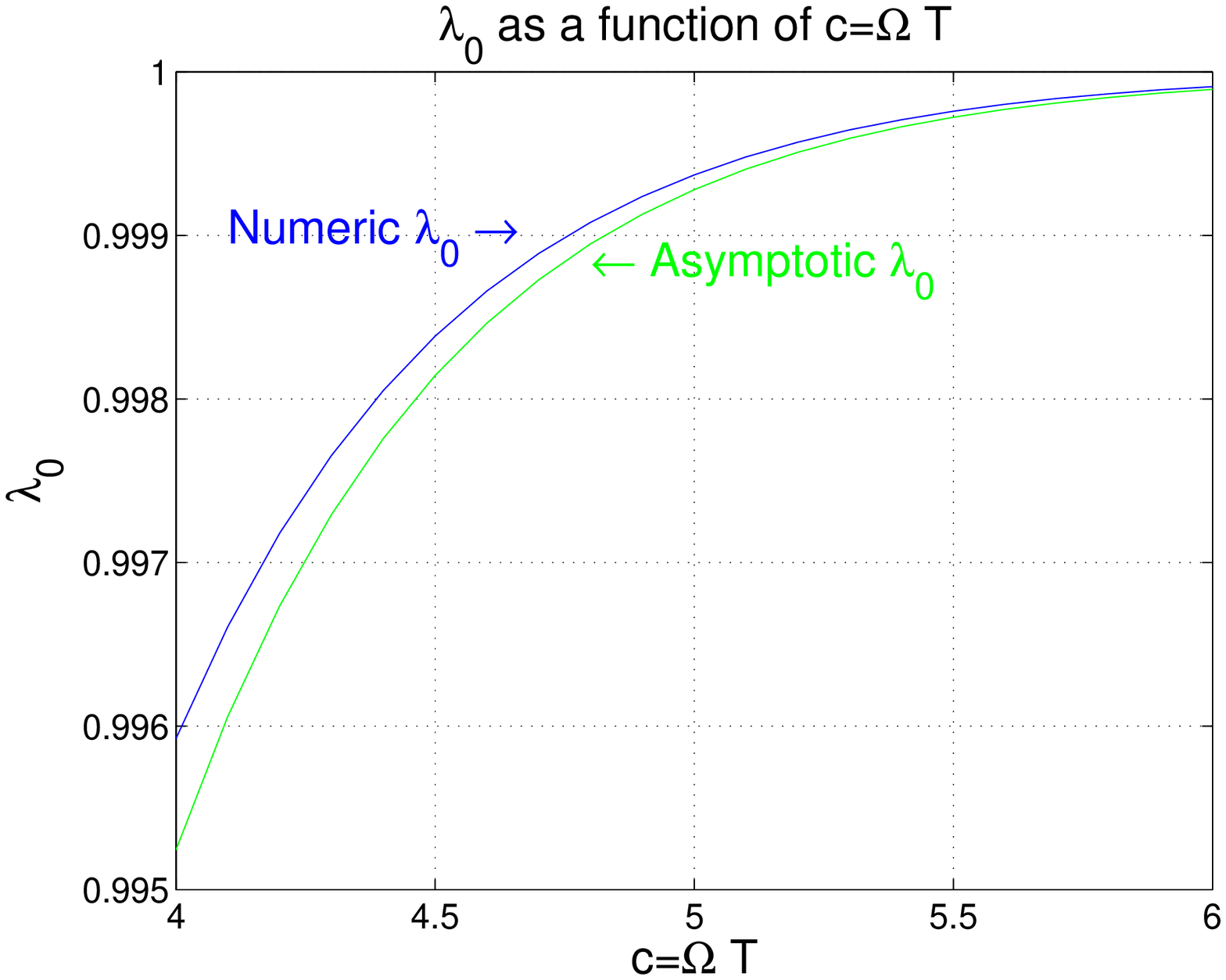}}
  }
    \centerline{
    \mbox{\includegraphics[scale=0.4]{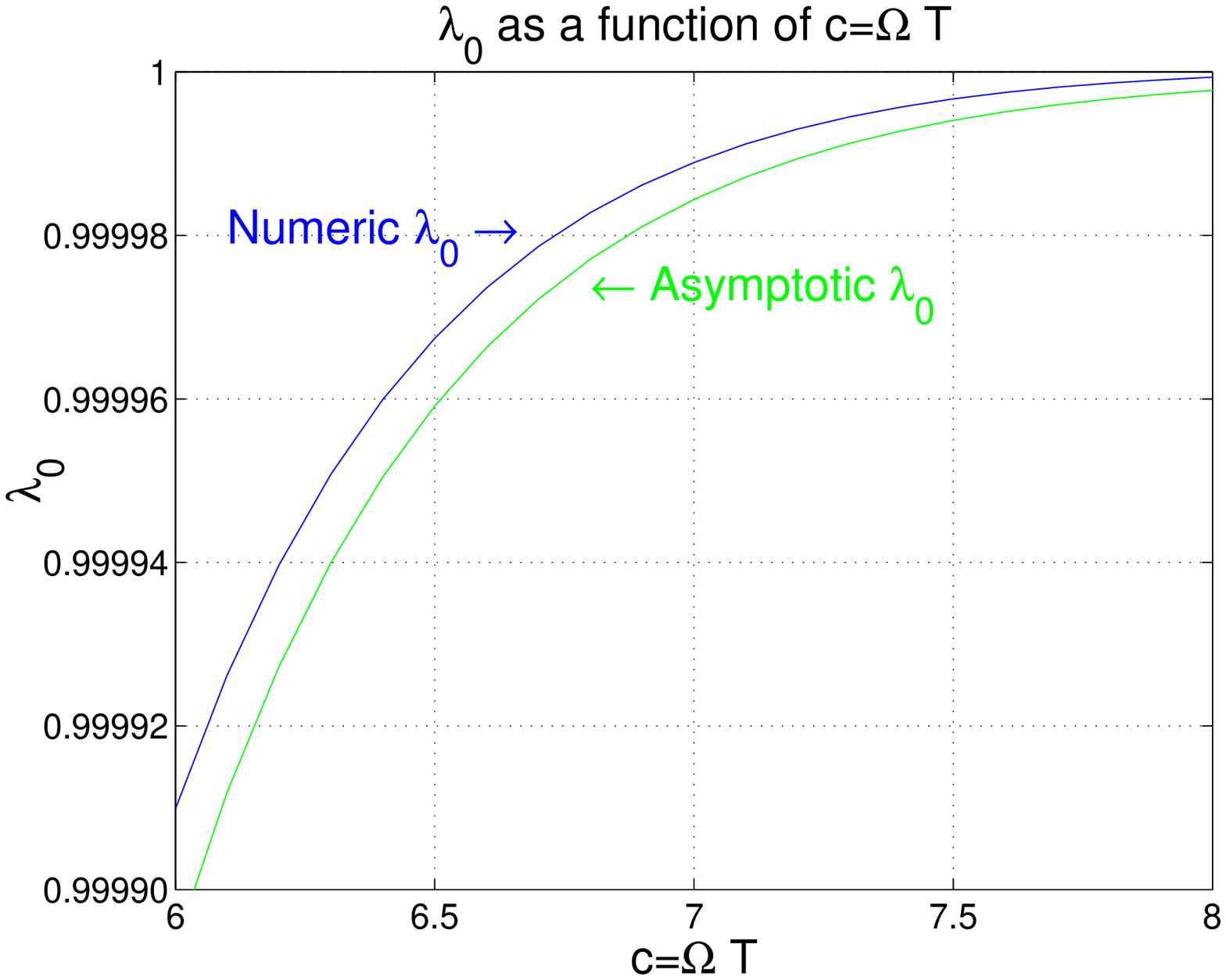}}
    \mbox{\includegraphics[scale=0.4]{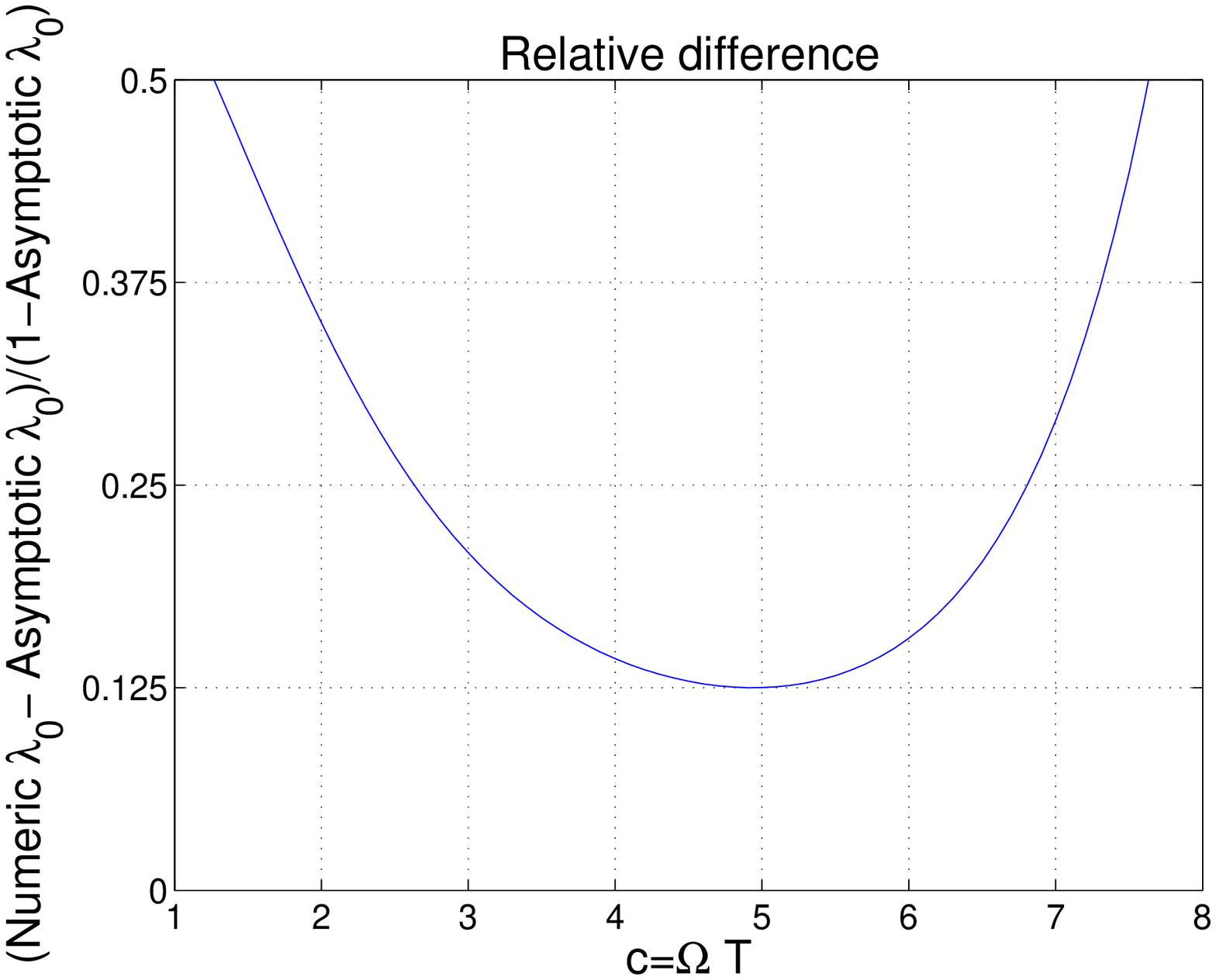}}
  }
  \caption{ $\lambda_0$ as a function of $c=\Omega T$ and the relative difference between the asymptotic formula ($1-4\pi^{1/2}c^{1/2}e^{-2c}$) and the numerical calculations .}
  \label{lambda_0_as_a_function_of_c_fig}
\end{figure}

LP \cite{LandauHJ_PollakHO1} mentioned that their theorem can be
used for describing the function $\phi(\alpha,\beta)$ that can be
used for writing the uncertainty principle in the form:

\begin{equation}
\label{uncertainty1}
 c=\Omega T\geq \phi(\alpha,\beta)
\end{equation}

or equivalently (to match Figure
\ref{Concentration_alpha2_possibility_maps_fig}; LP used the
concentration parameters $\alpha^2_T$ and $\beta^2_\Omega$. We
will discuss the parameters issue in section~\ref{sec5}):

 \begin{figure}[!h]
  \centerline{
    \mbox{\includegraphics[scale=0.4]{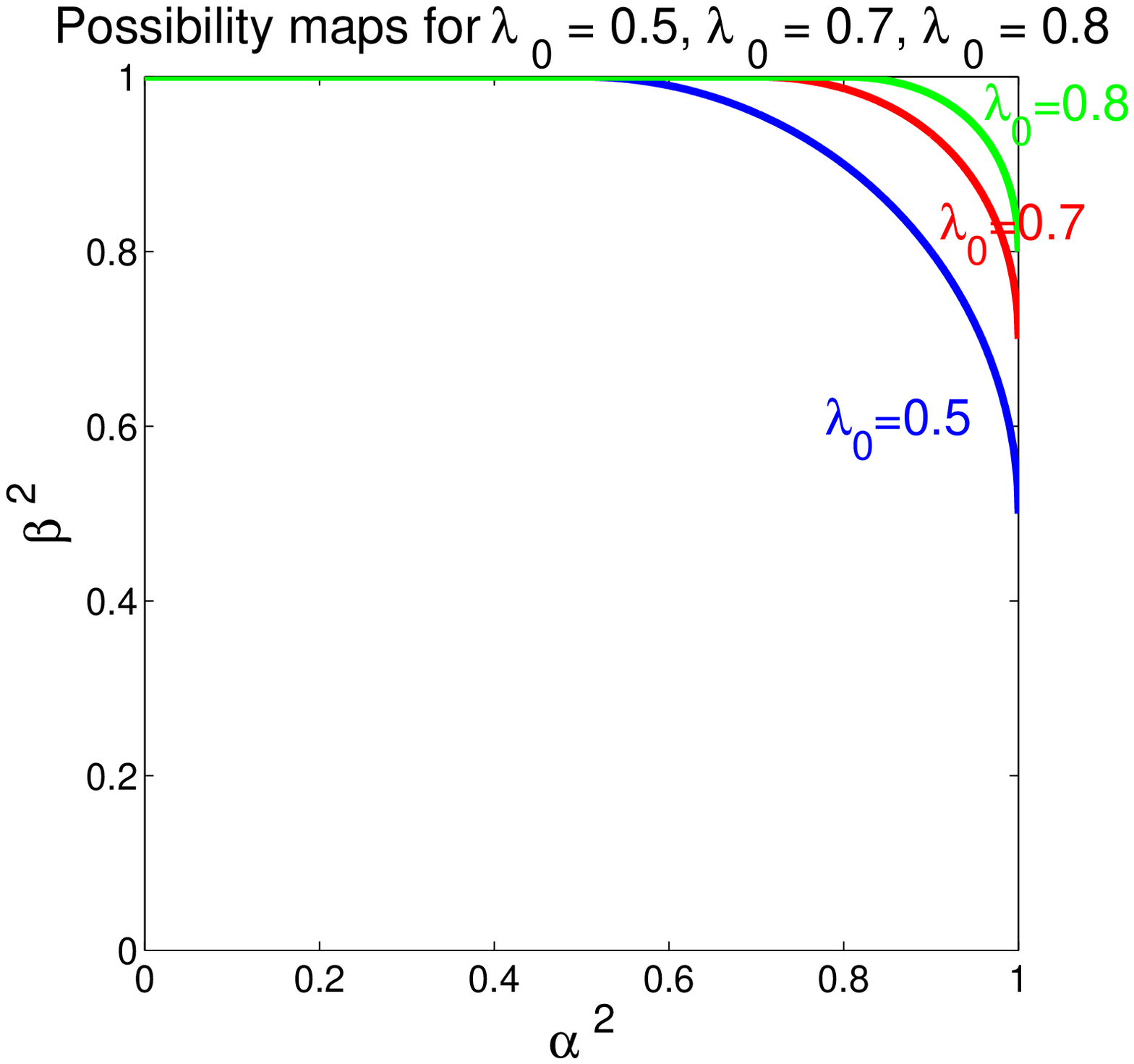}}
    \mbox{\includegraphics[scale=0.4]{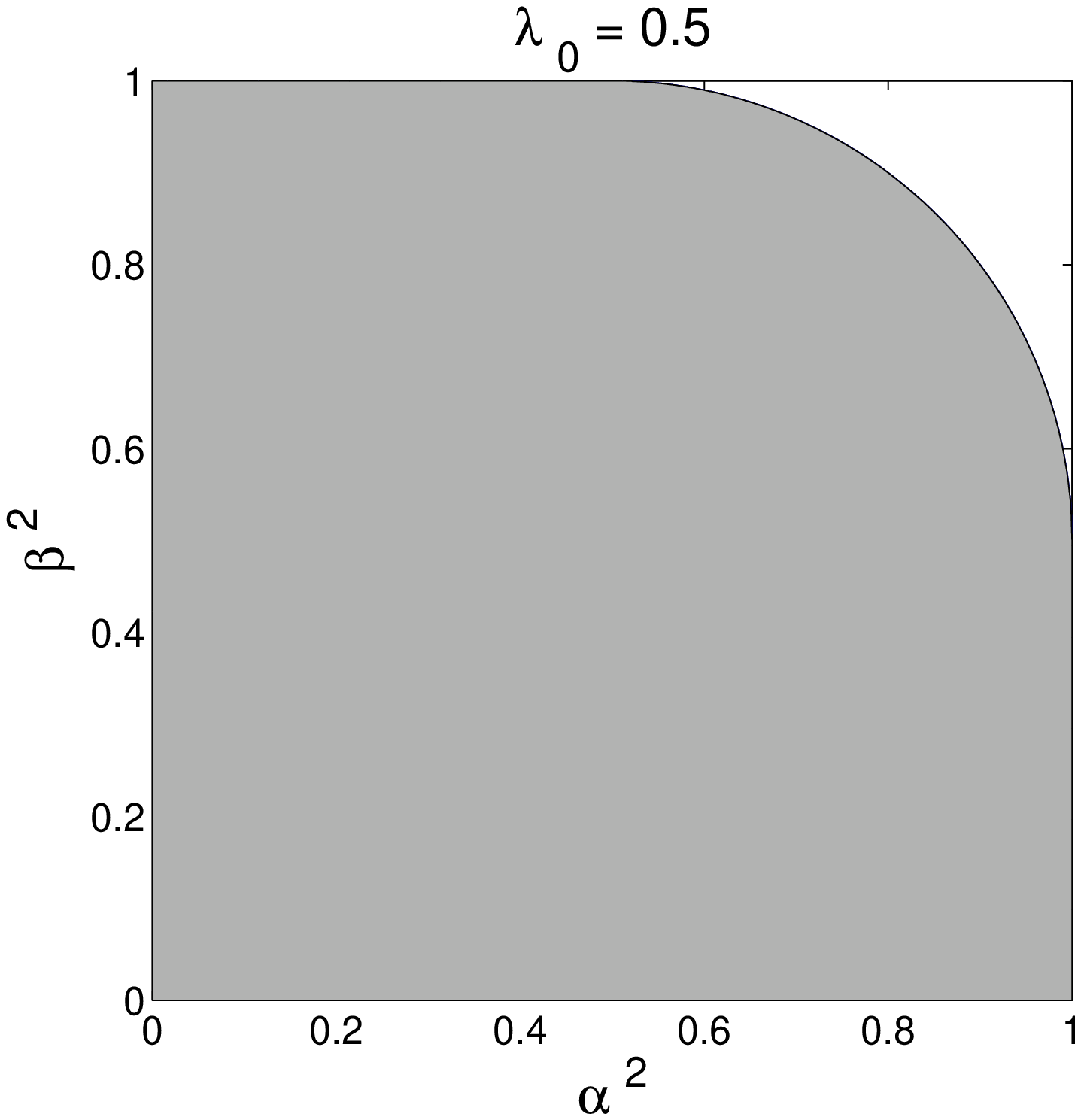}}
  }
  \centerline{
    \mbox{\includegraphics[scale=0.4]{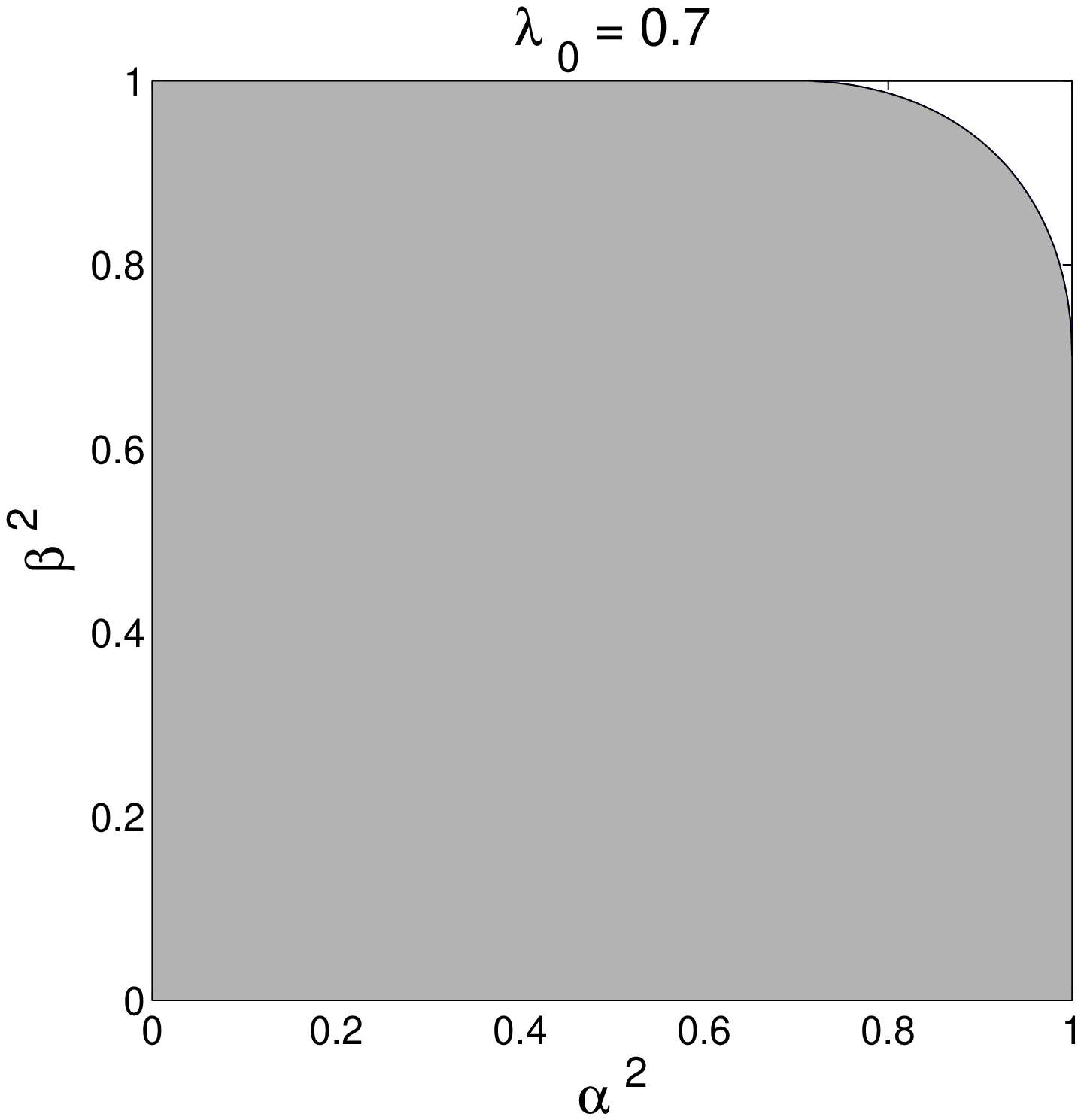}}
    \mbox{\includegraphics[scale=0.4]{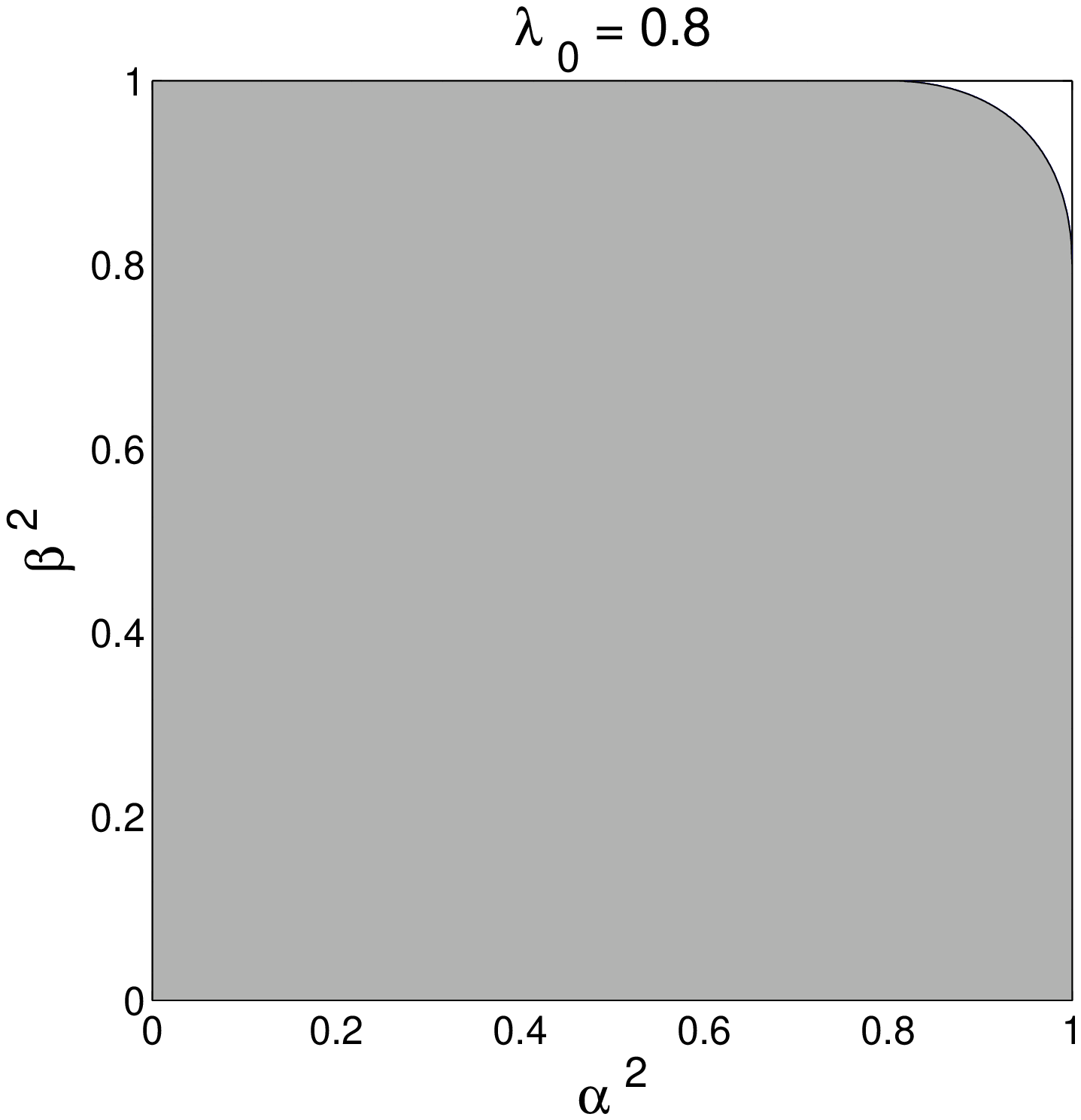}}
  }
  \caption{ Possibility maps for $\lambda_0=0.5, \lambda_0=0.7,
\lambda_0=0.8$ and the concentration parameters $\alpha^2$ and
$\beta^2$.}
 \label{Concentration_alpha2_possibility_maps_fig}
\end{figure}

\begin{equation}
\label{uncertainty2} c=\Omega T\geq
\Phi(\alpha^2,\beta^2)=\phi(\alpha,\beta)
\end{equation}

To the inequalities \ref{uncertainty1} and \ref{uncertainty2}
above and to their generalizations that will be introduced below
we will call $\textit{Generalized Uncertainty Principles}$.

LP used the concentrations $\alpha_T(f)$
\eqref{Time_concentration} and $\beta_\Omega(f)$
\eqref{Frequency_concentration} for formulating their generalized
uncertainty principle. The notion of $\textit{spreading}$ is in
some sense dual to the notion of $\textit{concentration}$. In
section \ref{sec4} we will consider uncertainty principles for
general weights using the spreading notion and describe the
properties of the weights, $w_1(x)$ and $w_2(\omega)$
(see~\eqref{Time_spreading1} and~\eqref{Frequency_spreading1}),
rigourously. In this section we define the time spreading
\eqref{Time_spreading1} and frequency spreading
\eqref{Frequency_spreading1} and use specific weights
\eqref{LPtime_weight} and \eqref{LPfrequency_weight} that fits to
LP Theorem. We will restate LP Theorem, using the spreadings
$\gamma(f,a)$ and $\zeta(f,b)$ (see below \eqref{Time_spreading1},
\eqref{Frequency_spreading1}, \eqref{LPtime_weight} and
\eqref{LPfrequency_weight}). Below we will redefine the terms
\textit{"possibility map"}, \textit{"possibility body"} and
\textit{"generalized uncertainty principle"} in an obvious way
which fits the notion of spreading instead of the notion of
concentration. We will use "the spreading language" from now on.

We define the time weight, $w_{1LP}(x)$, and the frequency weight,
$w_{2LP}(\omega)$ as follows:

\begin{equation}
\label{LPtime_weight} w_{1LP}(x)= \left\{
\begin{array}{ll}
         1 & \mbox{if $x \geq 1$};\\
        0 & \mbox{else}.\end{array} \right.
\end{equation}

\begin{equation}
\label{LPfrequency_weight} w_{2LP}(\omega)= \left\{
\begin{array}{ll}
         1 & \mbox{if $\omega \geq 1$};\\
        0 & \mbox{else}.\end{array} \right.
\end{equation}

We define the time spreading function as

\begin{defn}
\label{Time_spreading}
\begin{equation}
\label{Time_spreading1}
\gamma(f,a)=\frac{\big(\int_{-\infty}^\infty |f(x)|^2 S_a
w_1(x)dx\big)^\frac{1}{2}}{||f||}
\end{equation}
\end{defn}

and the frequency spreading as

\begin{defn}
\label{Frequency_spreading}
\begin{equation}
\label{Frequency_spreading1}
\zeta(f,b)=\frac{\big(\int_{-\infty}^\infty |\hat f(\omega)|^2 S_b
w_2(\omega)d\omega\big)^\frac{1}{2}}{||f||}
\end{equation}
\end{defn}

and see that in the case that $w_1(x)=w_{1LP}(x)$ and
$w_2(\omega)=w_{2LP}(\omega)$ we have
$\gamma^2(f,T)=1-\alpha^2_T(f)$ and
$\zeta^2(f,\Omega)=1-\beta^2_\Omega(f)$.

Therefore, defining $\Psi(\gamma^2,\zeta^2)$ and
$\psi(\gamma,\zeta)$ on

$D_{LP^*}=[0,1]\times[0,1]\setminus \{(0,0),(0,1),(1,0)\} $ by

\begin{equation*}
\Psi(\gamma^2,\zeta^2) = \Phi(1-\gamma^2,1-\zeta^2)
\end{equation*}

\begin{equation*}
\psi(\gamma,\zeta)=\Psi(\gamma^2,\zeta^2)
\end{equation*}

we get the LP uncertainty principle in the "spreading language":

\vspace{3 mm}

\begin{thm}[LP* Theorem]
\label{LP2theorem} There is a function $f$ such that $||f||=1$,
$\gamma(f,T)=\gamma$ and $\zeta(f,\Omega)=\zeta$ under the
following conditions and only under the following conditions:

1) $\gamma=0$ and $\sqrt{1-\lambda_0}\leq\zeta<1$

2) $0<\gamma\leq \sqrt{1-\lambda_0}$ and
$\cos^{-1}\sqrt{1-\gamma^2}+\cos^{-1}\sqrt{1-\zeta^2}\geq
\cos^{-1}\sqrt{\lambda_0}$

3) $\sqrt{1-\lambda_0}< \gamma <1$ and $0\leq\zeta\leq 1$

4) $\gamma=1$ and $ 0\leq 1$

where $\lambda_0$ is the largest eigenvalue of
\eqref{IntegralEquation}
\end{thm}
$\Box$

%$\diamondsuit$

and transforming equation~\eqref{uncertainty2} to an equivalent
general uncertainty principle in the "spreading language" we get:

\begin{equation}
\label{uncertainty3} c=\Omega T \geq
\Psi(\gamma^2,\zeta^2)=\psi(\gamma,\zeta)
\end{equation}

In section \ref{sec4} we will see that the "spreading point of
view" of LP theorem relates it to HPW uncertainty principle and to
general uncertainty principles that are described there. We will
redefine some of the terms that we used in the introduction, using
the "spreading language"

We redefine the "possible area of level c":

\begin{defn}
\label{LPpossible_area} The possible area of level c is the set:
$$M_c=M_{\Omega,T}=\bigcup_{f\in L^2}(\gamma(f,T),\zeta(f,\Omega))$$
\end{defn}

and redefine the "impossible area of level c", the complement of
$M_{\Omega,T}$ in $D_{LP^*}$:

\begin{defn}
\label{LPimpossible_area} The impossible area of level c is the
set:
$$M'_c=D_{LP^*}\setminus M_c$$
\end{defn}

 and call the pair $(M_c,M'_c)$ the possibility map of
level c.

We redefine the "possibility body" by:

\begin{defn}
\label{LPpossibility_body} The possibility body is the set:
\begin{equation}
 PB=\{ (\alpha,\beta,c) | (\alpha,\beta) \in M_c, c\in
(0,\infty)\}
\end{equation}
\end{defn}

 Notice that the "possible areas" in Figure~\ref{Concentration_alpha2_possibility_maps_fig}
(different $\lambda_0$'s corresponds to different $c$'s, see
Figure~\ref{lambda_0_as_a_function_of_c_fig}) are convex. Analytic
derivation of this fact is given below, in the example after
Theorem \ref{LPuncertainty}. To uncertainty inequalities of the
form \eqref{uncertainty3},we will call \textit{"Generalized
Uncertainty Principle"} (we redefined our definition that follows
\eqref{uncertainty2} to fit the "spreading language").

\begin{figure}[!h]
\label{Spreading_zeta_possibility_maps_fig}
  \centerline{
    \mbox{\includegraphics[scale=0.4]{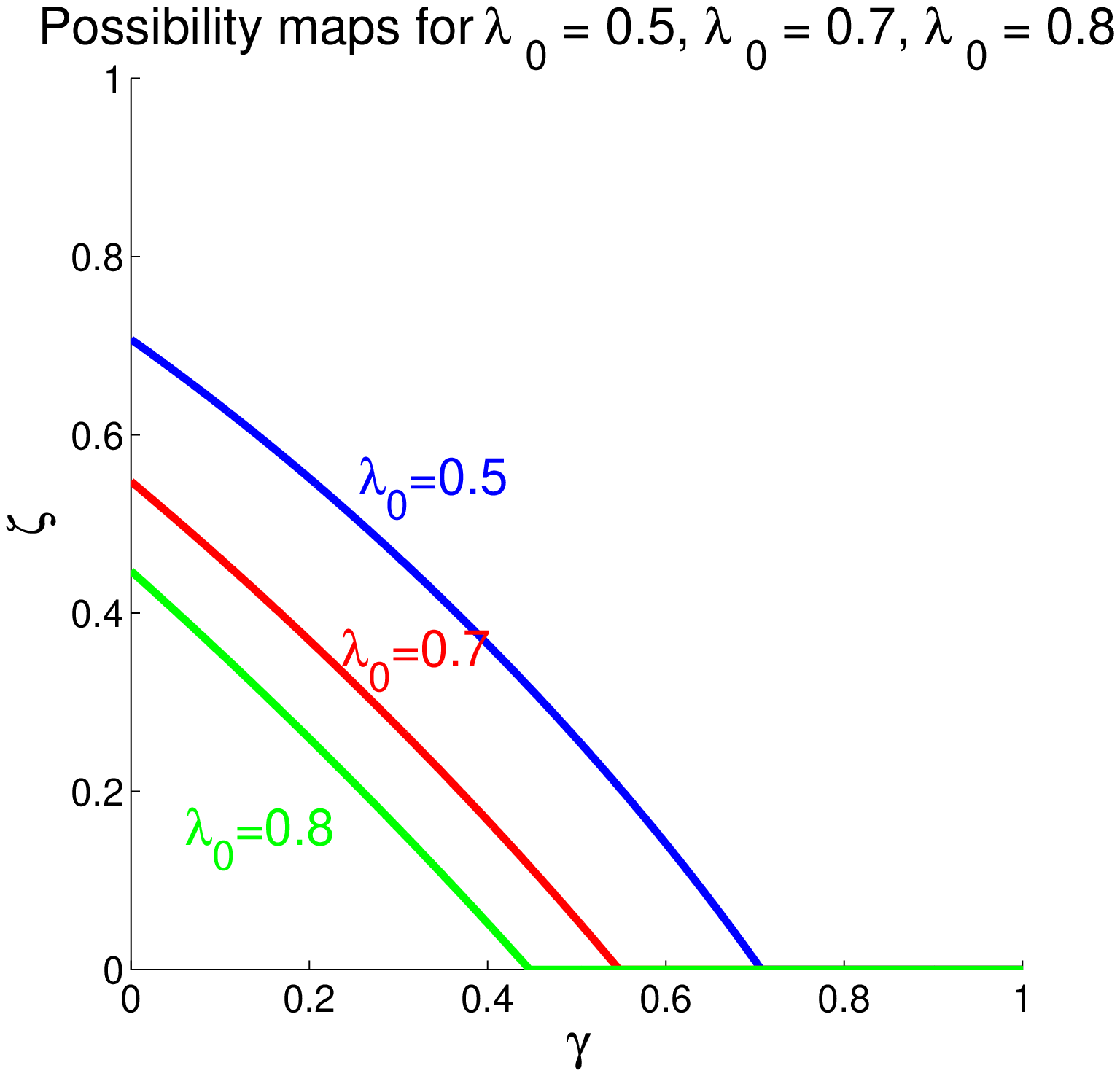}}
    \mbox{\includegraphics[scale=0.4]{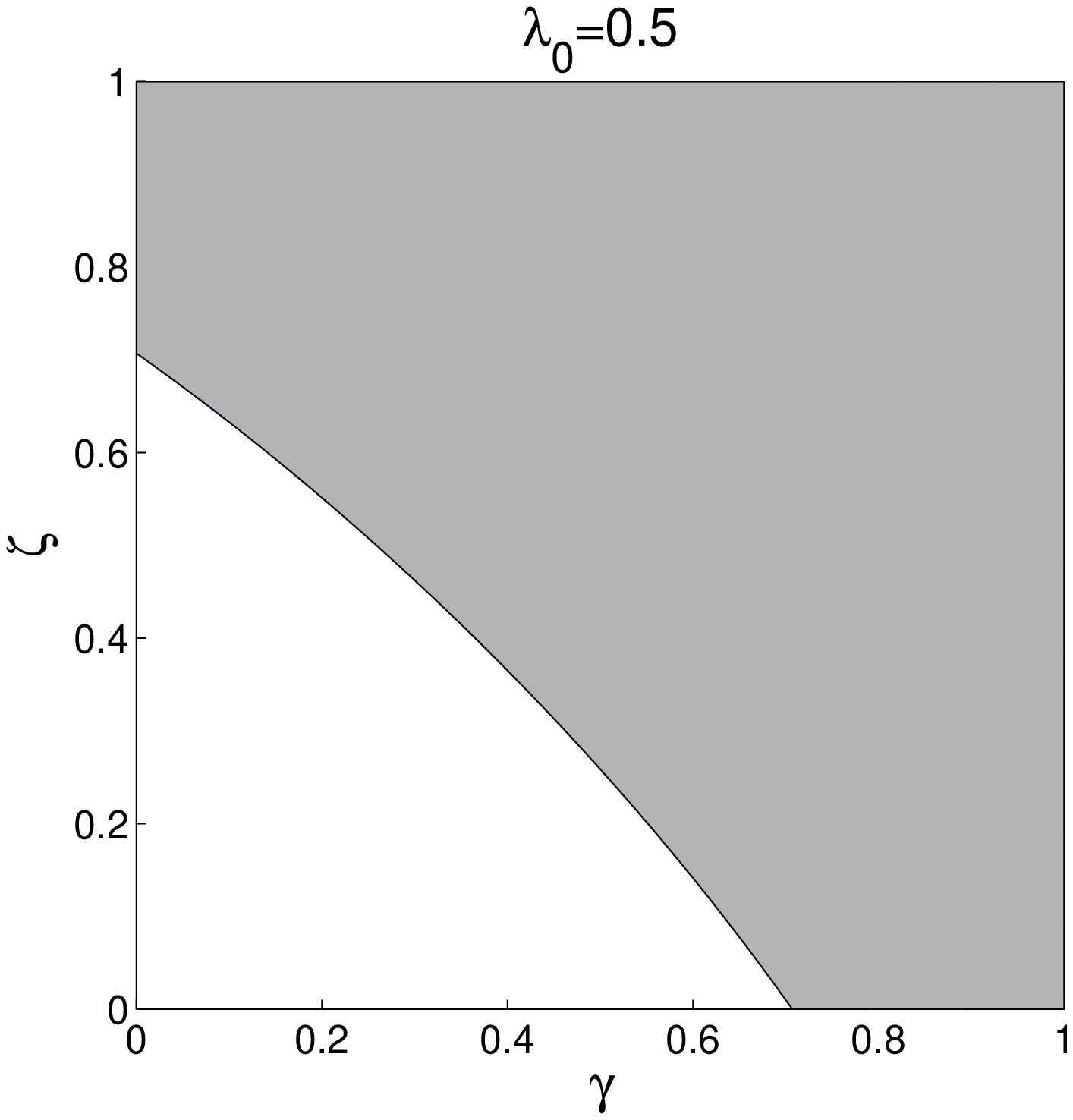}}
  }
  \centerline{
    \mbox{\includegraphics[scale=0.4]{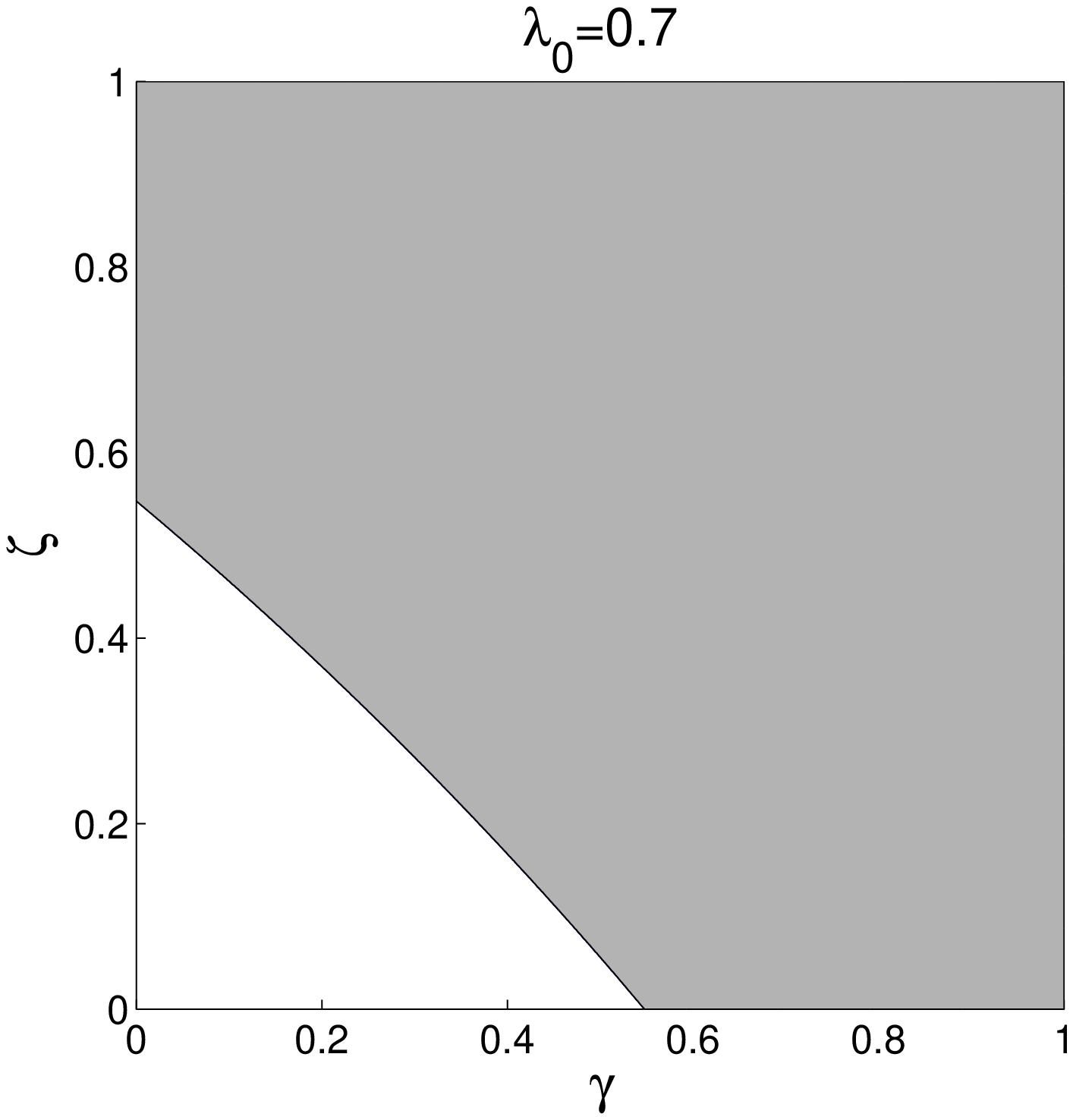}}
    \mbox{\includegraphics[scale=0.4]{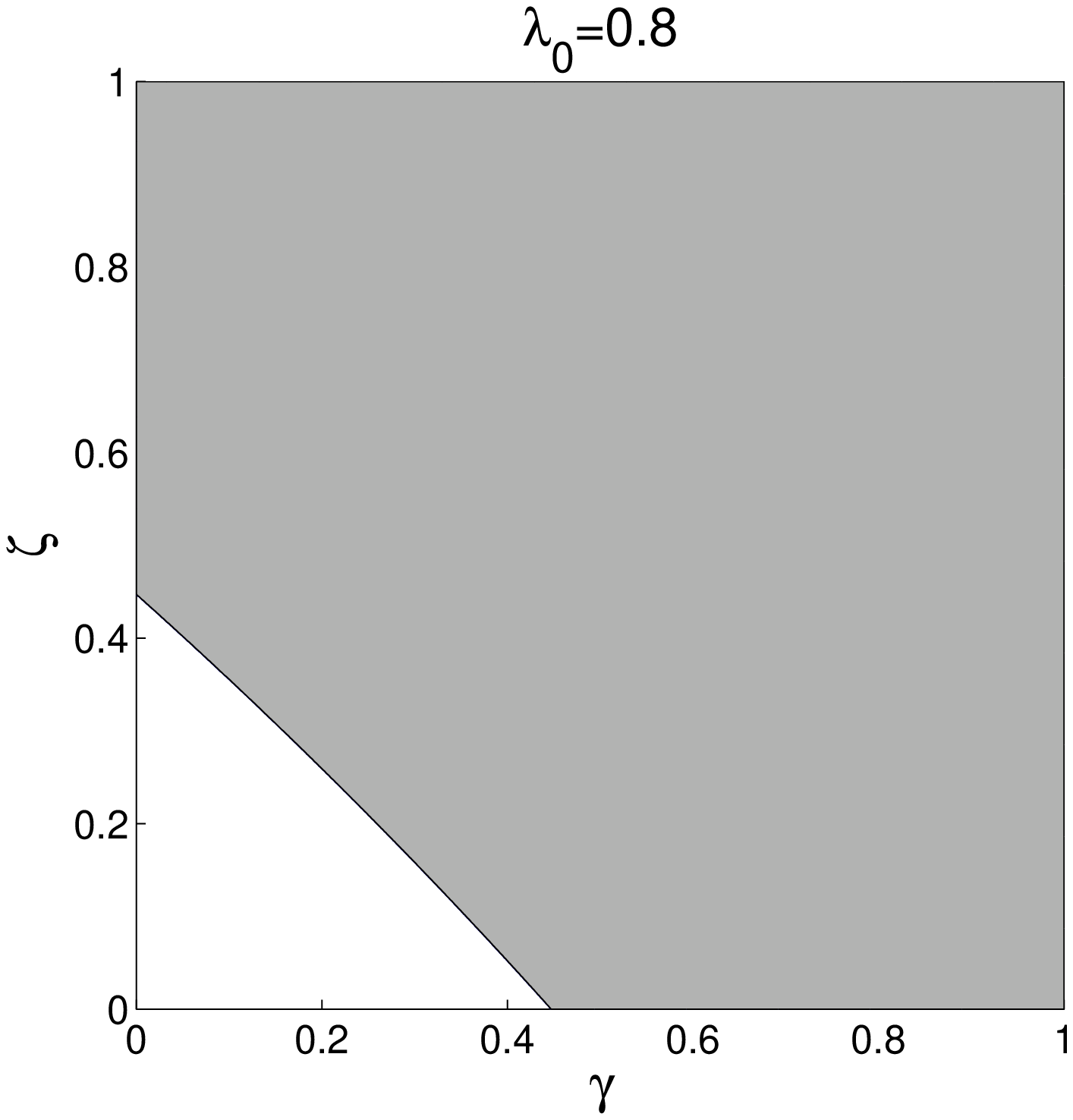}}
  }
  \caption{ Possibility map for $\lambda_0=0.5$, $\lambda_0=0.7$,
$\lambda_0=0.8$ using the coordinates (spreading parameters)
$\gamma$ and $\zeta$.}
\end{figure}

\vspace{7 mm}

$\psi$ is a non-increasing function of $\gamma$ for fixed $\zeta$
and a non-increasing function of $\zeta$ for fixed $\gamma$ since
$\lambda_0(c)$ is an non-decreasing function (see
Figure~\ref{lambda_0_as_a_function_of_c_fig})

Now we will derive classical uncertainty principles from the
general uncertainty principle of LP.

\vspace{7 mm}

\begin{thm}
\label{LPuncertainty}
 $\forall c=\Omega T>0$ we have:

\begin{equation}
\label{LPuncertainty_consequence} \sqrt{\int_{R\setminus [-T,T]}
|f(x)|^2dx}+ \sqrt{\int_{R\setminus [-\Omega ,\Omega]} |\hat
f(\omega)|^2 d\omega}\ge \sqrt{1-\lambda_0}||f||
\end{equation}

\end {thm}

\vspace{7 mm}

 %\prf
\begin{proof}
\vspace{3 mm}

We define $g(\gamma)=\cos^{-1}(\sqrt{1-\gamma^2})$ on $[-1,1]$.
$g(\gamma)$ is symmetric with respect to the point $0$ and on
$[0,1]$ we have: $g(\gamma)=\sin^{-1}(\gamma)$ on $[0, 1]$.
therefore $g(\gamma)$ is concave.

$\zeta$ as a function of $\gamma$, $0<\gamma\le
\sqrt{1-\lambda_0}$ defined implicitly by
$\cos^{-1}\sqrt{1-\gamma^2}+\cos^{-1}\sqrt{1-\zeta^2}=\cos^{-1}\sqrt{\lambda_0}$
is convex. To see that we write

$$g_2(\gamma,\zeta)=\cos^{-1}\sqrt{1-\gamma^2}+\cos^{-1}\sqrt{1-\zeta^2}=g(\gamma)+g(\zeta)$$

$g_2(\gamma,\zeta)$ is a concave function of two variables as a
sum of two concave functions:

\begin{align*}
&g_2(\lambda(x1,x2)+(1-\lambda)(x2,y2))=g_2((\lambda
x_1+(1-\lambda)x_2,\lambda y_1+(1-\lambda) y_2))\\
&=g(\lambda x_1+(1-\lambda)x_2)+g(\lambda y_1+(1-\lambda) y_2) <
\lambda g(x_1)+ (1-\lambda)g(x_2)+\\
&+\lambda g(y_1)+(1-\lambda)g(y_2)=\lambda g_2(x_1,y_1) +
(1-\lambda)g_2(x_2,y_2)
\end{align*}

from the symmetry of $g(\gamma)$ it follows that
$g_2(\gamma,\zeta)$ is symmetric with respect to the axes.
Therefore $\zeta$ as a function of $\gamma$, $0<\gamma\le
\sqrt{1-\lambda_0}$ defined implicitly by
$\cos^{-1}\sqrt{1-\gamma^2}+\cos^{-1}\sqrt{1-\zeta^2}=\cos^{-1}\sqrt{\lambda_0}$
is convex.

We have plotted $\zeta$ as a function of $\gamma$ in
Figure~\ref{Spreading_zeta_possibility_maps_fig}.

 \begin{figure}[!h]
  \centerline{
    \mbox{\includegraphics[scale=0.4]{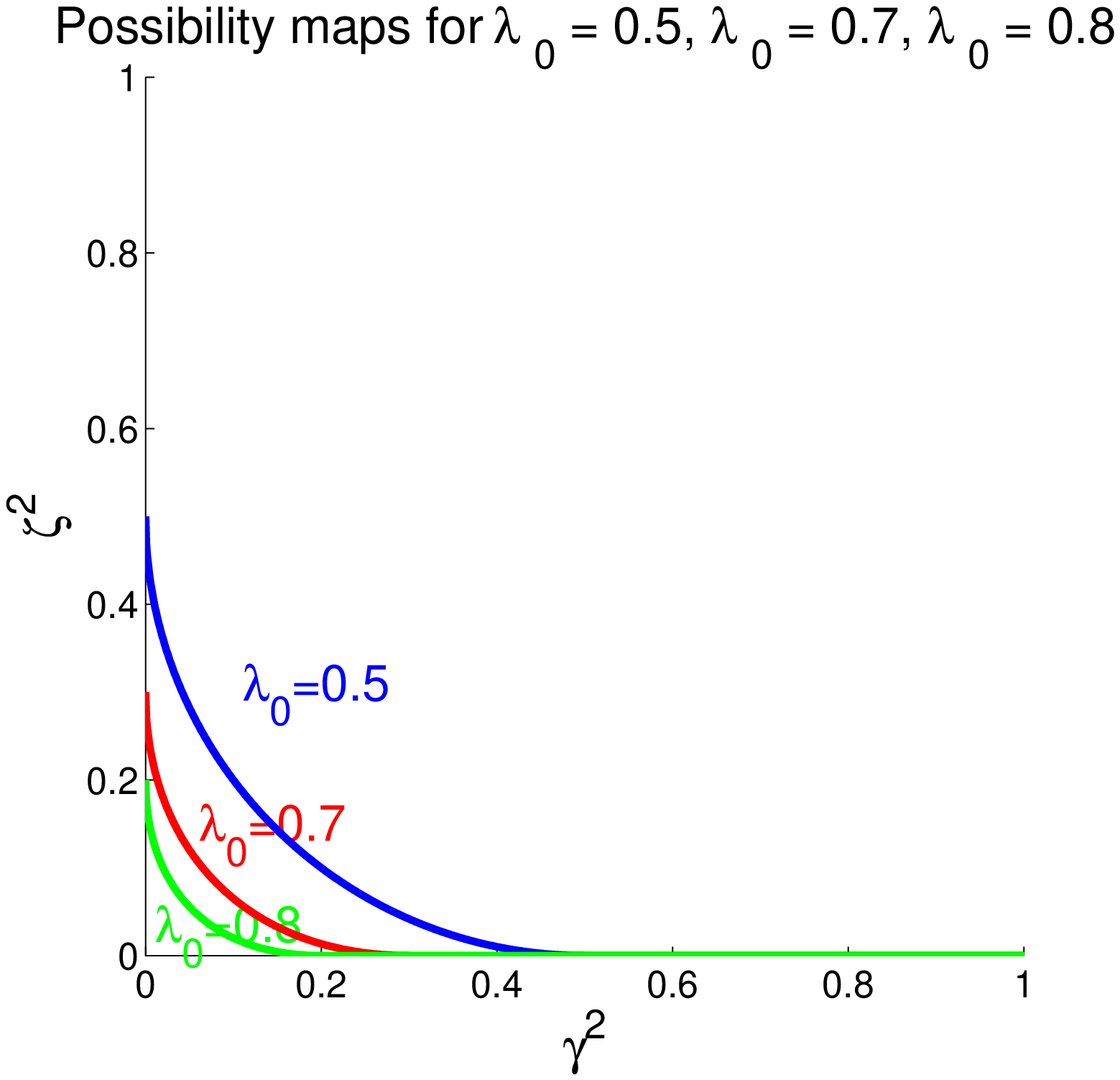}}
    \mbox{\includegraphics[scale=0.4]{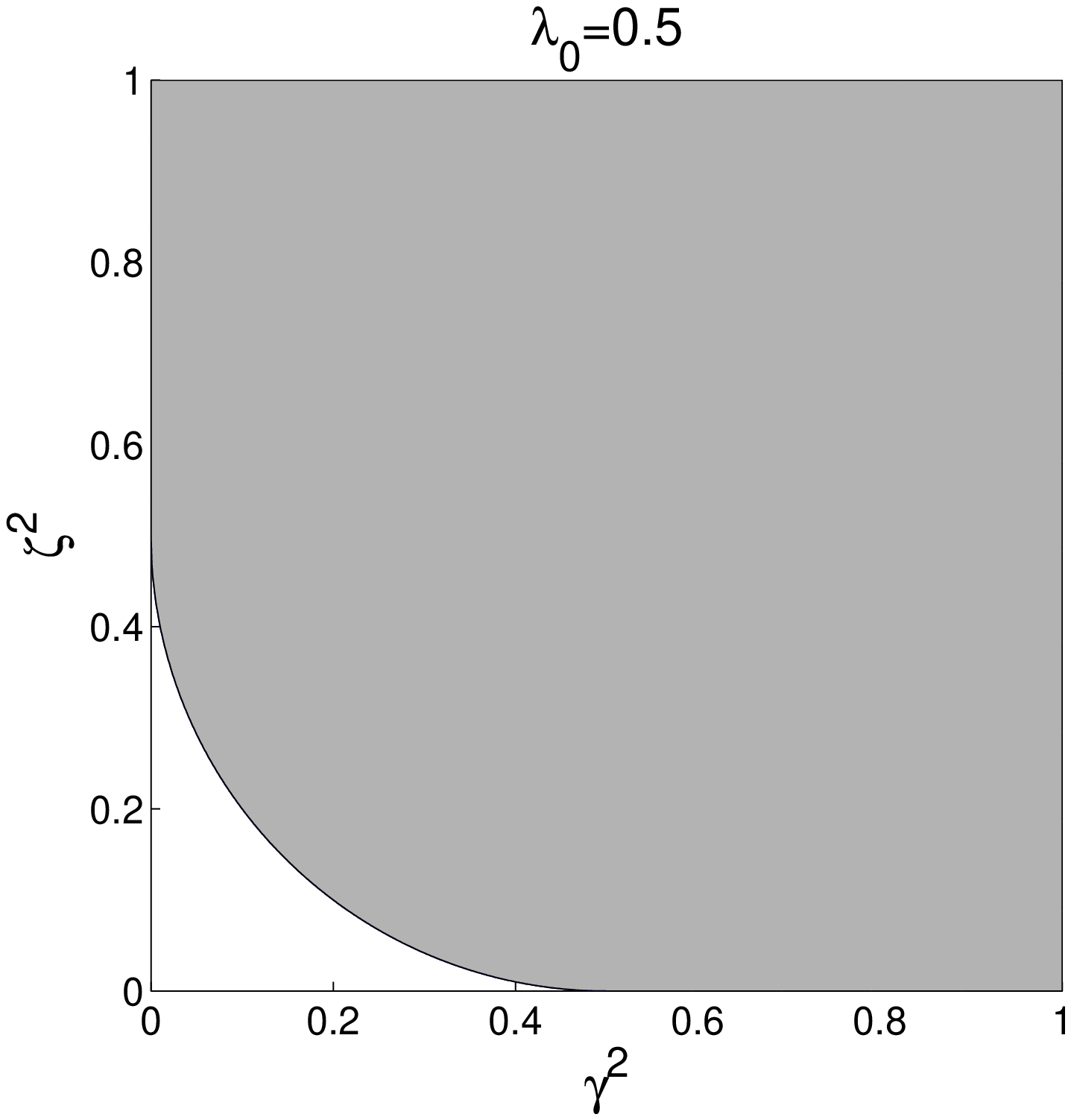}}
  }
  \centerline{
    \mbox{\includegraphics[scale=0.4]{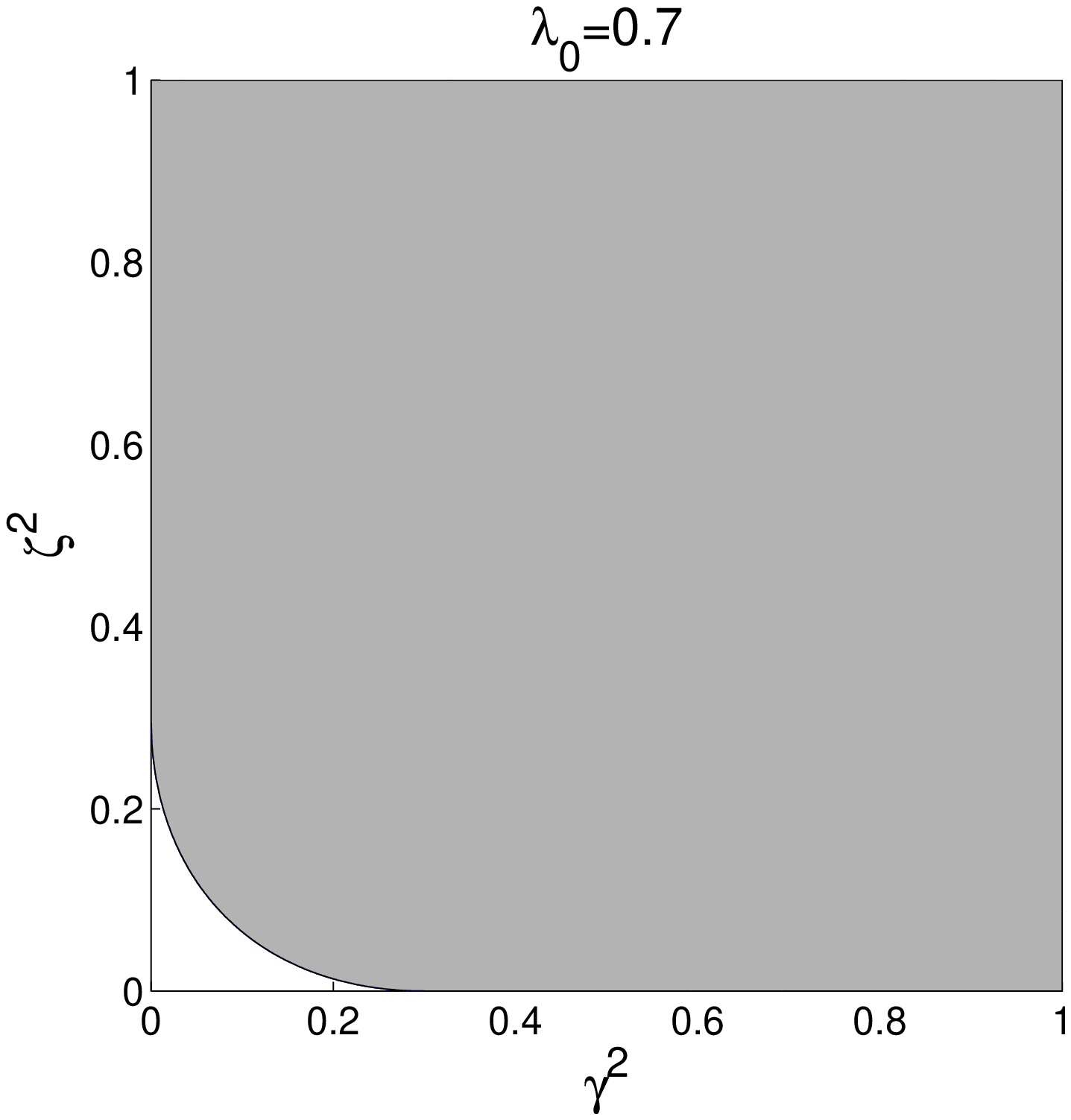}}
    \mbox{\includegraphics[scale=0.4]{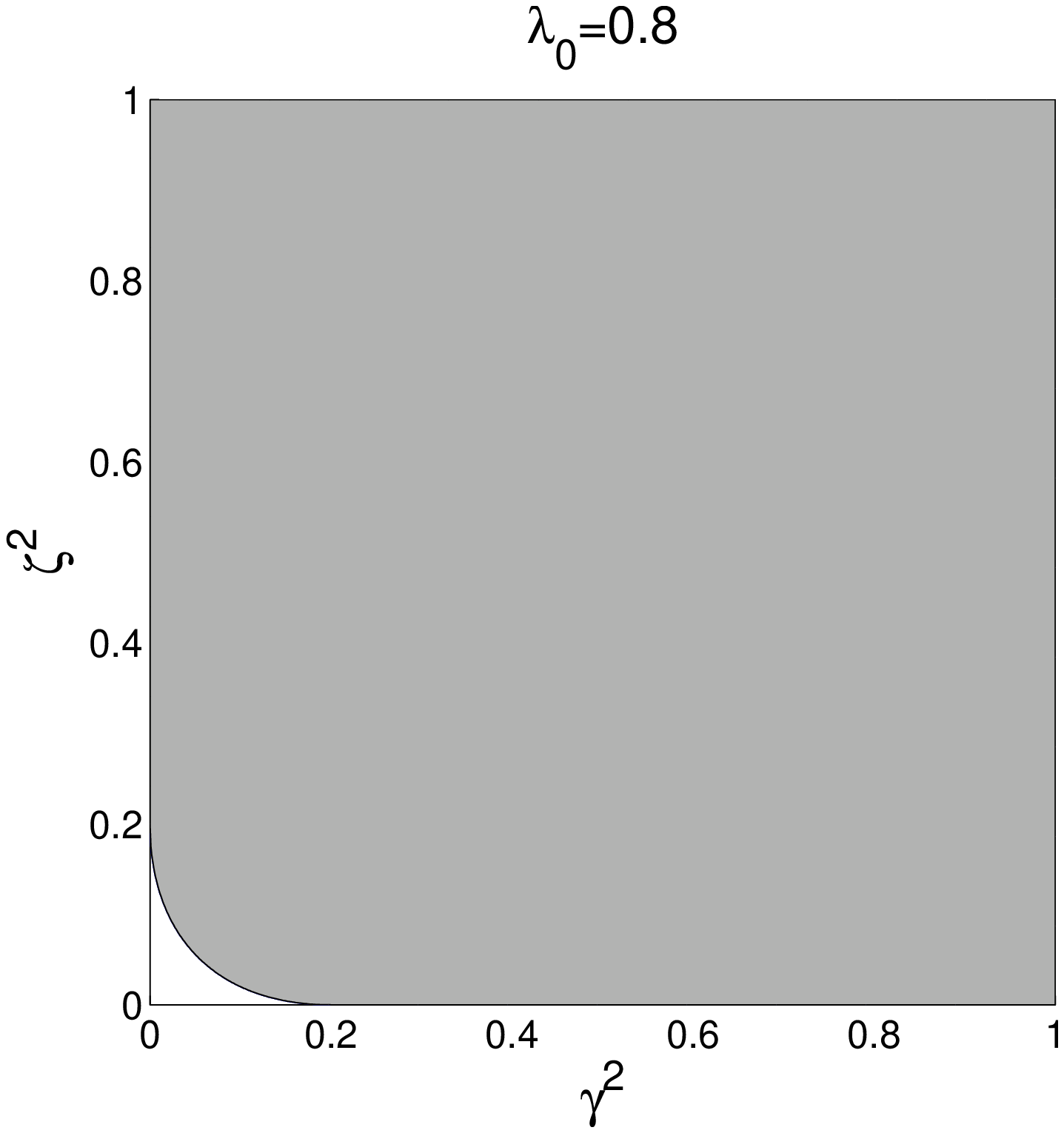}}
  }
  \caption{ Possibility map for $\lambda_0=0.5$, $\lambda_0=0.7$,
$\lambda_0=0.8$ using the coordinates (spreading parameters)
$\gamma^2$ and $\zeta^2$.}
  \label{Spreading_possibility_maps_alpha2_fig}
\end{figure}

thus the maximum of $\zeta+\gamma$ is at the two edge points
$(\sqrt{1-\lambda_0},0)$, $(0,\sqrt{1-\lambda_0})$ and we have the
following uncertainty principle:

\begin{equation}
\frac{\sqrt{\int_{R\setminus  [-\frac{T}{2},\frac{T}{2}]}
|f(x)|^2dx}}{||f||}+ \frac{\sqrt{\int_{R\setminus [-\Omega
,\Omega]} |\hat f(\omega)|^2 d\omega}}{||f||} \ge
\sqrt{1-\lambda_0}
\end{equation}

which implies \eqref{LPuncertainty_consequence}.

\end{proof}

In Theorem \ref{LPuncertainty} above we showed that $M_c$, $c>0$
is a non-convex set (we showed that $\zeta$ as a function of
$\gamma$ is convex, which is equivalent). We would like to remark
that if we use a different set of natural parameters (such as
$\zeta^2$ and $\gamma^2$, see
Figure~\ref{Spreading_possibility_maps_alpha2_fig}) then $M_c$,
where $c$ belongs to some interval, may be a convex set.

 As an example we show that using the parameters $\zeta^2$ and $\gamma^2$, $M_c$, where $\lambda_0(c)>\frac{1}{2}$, is a convex set:
  $M_c$ is convex iff  $\zeta^2$ as a function of
$\gamma^2$, $0<\gamma^2\le 1-\lambda_0$ defined implicitly by
$\cos^{-1}\sqrt{1-\gamma^2}+\cos^{-1}\sqrt{1-\zeta^2}=\cos^{-1}\sqrt{\lambda_0}$
is concave. To see that we write $\beta^2$ as a function of
$\alpha^2$

\begin{equation}
\beta^2(\alpha^2)=\cos^2(\cos^{-1}\sqrt{\lambda_0}-\cos^{-1}\sqrt{\alpha^2})
\end{equation}

we denote
$h(\alpha^2)=\cos^{-1}\sqrt{\lambda_0}-\cos^{-1}\sqrt{\alpha^2}$
and differentiate $\beta^2(\alpha^2)$ ($\beta^2$ with respect to
$\alpha^2$ twice) to get

\begin{equation}
(\beta^2)'(\alpha^2)=\frac{\cos(h(\alpha^2))sin(h(\alpha^2))}{\sqrt{\alpha^2(1-\alpha^2)}}
\end{equation}

\begin{equation}
(\beta^2)''(\alpha^2)=\frac{1}{2}\frac{sin^2(h(\alpha^2))}{\sqrt{\alpha^2(1-\alpha^2)}}-\frac{1}{2}\frac{\cos^2(h(\alpha^2))}{\sqrt{\alpha^2(1-\alpha^2)}}+
\end{equation}

\begin{equation}
+\frac{1}{2}\frac{\cos(h(\alpha^2))sin(h(\alpha^2))}{{(\alpha^2)}^{\frac{3}{2}}(1-\alpha^2)^{\frac{1}{2}}}-\frac{1}{2}\frac{\cos(h(\alpha^2))sin(h(\alpha^2))}{{(\alpha^2)}^{\frac{1}{2}}(1-\alpha^2)^{\frac{3}{2}}}
\end{equation}

for $\lambda_0\geq\frac{1}{2}$ since $sin^2(h(\alpha^2))\leq
\cos^2(h(\alpha^2))$ and
${(\alpha^2)}^{\frac{3}{2}}(1-\alpha^2)^{\frac{1}{2}} \geq
{(\alpha^2)}^{\frac{1}{2}}(1-\alpha^2)^{\frac{3}{2}}$ we get that
$\beta^2(\alpha^2)\leq 0$. $\zeta^2(\gamma^2)$ is concave since
$\beta^2(\alpha^2)$ is convex.

\vspace{3 mm}

One can get different inequalities by minimizing different
functions on the possible area. Of course, one can use different
coordinates systems to work with for obtaining different
inequalities. As an illustration we will use the coordinates
$\zeta^2$ and $\gamma^2$. In this case we will not get a new
uncertainty principle. We will get a weaker uncertainty principle
then \ref{LPuncertainty_consequence}, but the illustration is
instructive:

From symmetry with respect to the line $\zeta^2=\gamma^2$ and
concavity of the function $\zeta^2(\gamma^2)$, the maximum of the
function $\gamma^2+\zeta^2$ on the possibility map is attained at
the point
$\gamma^2=\zeta^2=1-\cos^2\frac{\cos^{-1}\sqrt{\lambda_0}}{2}$.
(We used the equation
$2\cos^{-1}\sqrt{1-\gamma^2}=\cos^{-1}\sqrt{\lambda_0}$).

and we get:

\begin{equation}
\frac{\int_{-\infty}^\infty |f(x)|^2 S_{T}
w_{1LP}(x)dx}{||f||^2}+\frac{\int_{-\infty}^\infty |\hat
f(\omega)|^2 S_{\Omega} w_{2LP}(\omega)d\omega}{||f||^2}\ge
2(1-\cos^2\frac{\cos^{-1}\sqrt{\lambda_0}}{2})
\end{equation}

using the identity:

$$ \cos\frac{u}{2}=\sqrt{\frac{1+\cos u}{2}}$$

we get the uncertainty principle:

\begin{equation}
\label{LPuncertainty_weak_consequence} \sqrt{\int_{R\setminus
[-T,T]} |f(x)|^2dx+
     \int_{R\setminus [-\Omega ,\Omega]} |\hat f(\omega)|^2
d\omega} \ge \sqrt{1-\lambda_0}||f||
\end{equation}

As we mentioned the uncertainty principle
\eqref{LPuncertainty_weak_consequence} is weaker then
\eqref{LPuncertainty_consequence}. In fact it is a consequence of
\eqref{LPuncertainty_consequence} using the inequality
$$\sqrt{2}\sqrt{a+b}\ge \sqrt{a}+\sqrt{b}$$

Remark: We wanted to point out the convexity of $\zeta(\gamma)$
and the concavity of $\zeta^2(\gamma^2)$. If one is only
interested in finding the inequalities, there may be easier ways
to get those.

\section{uncertainty principles for general weights}
\label{sec4}

In this section we will generalize the example from section
\ref{sec3} to more general weights. We denote the time weight by
$w_1(x)$ and the frequency weight by $w_2(\omega)$. We will use
the notations:

$${w_1}_a(x)=S_aw_1(x)=w_1(xa^{-1})$$

where $a$ is the called the time weight scaling parameter; and

$${w_2}_b(\omega)=S_bw_2(\omega)=w_2(\omega b^{-1})$$

where $b$ is the called the frequency weight scaling parameter. We
recall the definitions of time spreading and frequency spreading
(see definitions \ref{Time_spreading} and
\ref{Frequency_spreading}:

\vspace{3 mm}

The time spreading:

\begin{equation}
\gamma(f,a)=\frac{\big(\int_{-\infty}^\infty
{w_1}_a(x)|f(x)|^2dx\big)^\frac{1}{2}}{||f||}
\end{equation}

\vspace{3 mm}

The frequency spreading:

\begin{equation}
\frac{\zeta(f,b)=\big(\int_{-\infty}^\infty {w_2}_b(\omega)|\hat
f(\omega)|^2d\omega \big)^\frac{1}{2}}{||f||}
\end{equation}

In section \ref{sec3} we used specific weights $w_{1LP}(x)$ (see
\eqref{LPtime_weight}) and $w_{2LP}(\omega)$ (see
\eqref{LPfrequency_weight}). In this section we will use general
weights that posses mild requirements.

\vspace{3 mm}

\begin{defn}
\label{realizable_point4}
 A point $(p,q,a,b)\in R^{4^+}$ is called realizable iff there is a
function $f\in L_2$ such that: $p=\gamma(f,a)$,
$q=\zeta(f,b),a>0,b>0$
\end{defn}

\vspace{3 mm}

\begin{lem}
\label{realizable_point_lemma}
 A point $(p,q,a_1,b_1)\in R^{4+}$ is
realizable iff the point $(p,q,k a_1,\frac{b_1}{k}), k>0$ is
realizable
\end{lem}
\vspace{3 mm}

\begin{proof}
 It is enough to show that If

\vspace{1 mm}

$||f||=1$, $p=\gamma(f,a_1)$, $q=\zeta(f,b_1)$, $a_2=ka_1$,
$b_2=\frac{b_1}{k}$

\vspace{1 mm}

then the function $f_k(x)=\frac{1}{\sqrt{k}}f(\frac{x}{k})$
satisfies:

\vspace{1 mm}

$||f_k||=1$, $p=\gamma(f_k,a_2)$, $q=\zeta(f_k,b_2)$

\begin{eqnarray*}
\gamma(f_k,a_2) & = & (\int_{-\infty}^\infty w_1(\frac{x}{a_2})|\frac{1}{\sqrt{k}}f(\frac{x}{k})|^2dx)^\frac{1}{2} \\
             & = & (\int_{-\infty}^\infty w_1(\frac{y}{a_1})\frac{1}{k}|f(y)|^2kdy)^\frac{1}{2} \\
             & = & \gamma(f,a_1)=p
\end{eqnarray*}

where $y=\frac{x}{k}$

We use
$\widehat{\frac{1}{\sqrt{k}}f(\frac{x}{k})}=\sqrt{k}\hat{f}(k\omega)$
and see:

\begin{eqnarray*}
\zeta(f_k,b_2) & = & (\int_{-\infty}^\infty w_2(\frac{\omega}{b_2})|\sqrt{k}\hat{f}(k\omega)|^2d\omega)^\frac{1}{2} \\
             & = & (\int_{-\infty}^\infty w_2(\frac{v}{b_1})k|\hat{f}(v)|^2\frac{dv}{k})^\frac{1}{2} \\
             & = & \zeta(f,a_1)=q
\end{eqnarray*}

where $v=k\omega$

\end{proof}
%$\diamondsuit$

 \vspace{7 mm}

Note that Lemma \ref{realizable_point_lemma} is steal correct if
we use a change of coordinates from $R^{4^+}$ to $R^{4^+}$ of the
form $(p,q,a,b)\rightarrow (s(p,q),t(p,q),a,b)$. We will use
specific change of variables in Section \ref{sec5}.

Lemma \ref{realizable_point_lemma} indicates that the relevant
parameter is the time weight scaling parameter $\times$ the
frequency weight scaling parameter; So, we can define realizable
points in $R^{3+}$ instead of in $R^{4+}$:

 \vspace{3 mm}

\begin{defn}
\label{realizable_point3}
 A point $(p,q,c)\in R^{3+}$ is called realizable iff there is a
function $f\in L_2$ such that, $||f||=1$, $p=\gamma(f,a)$,
$q=\zeta(f,b)$ and $c=a b$.
\end{defn}

\vspace{3 mm}

\begin{defn}
\label{possibility_body} We will call the set of realizable points
in $R^{3+}$ "possibility body" and denote it by
$PB_{\gamma^1,\zeta^1}$.
\end{defn}

Note that we may think of the set $PB_{\gamma^1,\zeta^1}$ as the
set:

\begin{equation*}
\bigcup_{a,b>0,f\in L_2} (\gamma(f,a),\zeta(f,b),a ,b)
\end{equation*}

where we identify points such that $ab=c$.

The meaning of the subindexes will be clear in Section \ref{sec5},
since definitions \ref{realizable_point4}, \ref{realizable_point3}
and \ref{possibility_body} are special cases of the more general
definitions \ref{realizable_point4_of_oredr_m},
\ref{realizable_point3_of_order_m} and
\ref{possibility_body_of_order_m} in section \ref{sec5}. In the
following it also will become clear that the possibility body PB
defined in Definition \ref{LPpossibility_body}. is similar to
$PB_{\gamma^1,\zeta^1}$ when we use LP-weights
\eqref{LPtime_weight} and \eqref{LPfrequency_weight}.

\vspace{3 mm}

We will use the following type of weights in the theorem. The
weights in LP result , \eqref{LPtime_weight} and
\eqref{LPfrequency_weight}, are pointwise limit of weights of type
1.

\begin{defn}[Weight of type 1]

A weight $w(x)$ is of type 1 if and only if:

\begin{itemize}

\item[a)] $w(0)=0$.

\item[b)] $w(x)$ is a continuous function.

\item[c)] $w(x)$ is an even function i.e. $w(x)=w(-x)$.

\item[d)] $w(x)$ is strictly increasing on $[0,\infty]$ i.e. $\forall
0\leq x<y$ $w(x)<w(y)$.

\item[e)] $w(x)$ tends to a finite number as $x$ tends to $\infty$ i.e.

$\lim_{x\rightarrow \infty}w(x) = L<\infty$

\end{itemize}
\end{defn}

\vspace{3 mm}

Restricting ourselves to  weights of type 1 we get the following
generalized uncertainty principle:

\begin{thm}[General Uncertainty Theorem For Weights Of Type 1]

Let $w_1(x)$ and $w_2(\omega)$ be weights of type 1 where
$$\lim_{x\rightarrow \infty}w_1(x) = P$$
and
$$\lim_{x\rightarrow \infty}w_2(x) = Q$$

Then the possibility body,$PB_{\gamma^1,\zeta^1}$, is defined, up
to a set of measure zero, by a generalized uncertainty inequality
of the form:

\begin{equation}
\label{uncertainty_finite_lim} ab=c\geq
\psi_{w_1(x),w_2(\omega)}(\gamma,\zeta)
\end{equation}

where $\psi$ is defined on the open square $N=(0,\sqrt{P})\times
(0,\sqrt{Q})$

$\psi_{w_1(x),w_2(\omega)}(\gamma,\zeta)$ is a non-increasing
function of $\gamma$ for fixed $\zeta$ and a non-increasing
function of $\zeta$ for fixed $\gamma$

If $w_1(x)=w_2(x)$ then $\psi_{w_1(x),w_2(\omega)}(\gamma,\zeta)$
is symmetric \label{generalized_uncertainty_theorem1}
\end{thm}

%\prf
\begin{proof}

The proof has 4 steps:

Step 1:

 If a point $(c,\gamma_1,\zeta)$ is realizable then all
points $(c,\gamma_2,\zeta)$, $\gamma_1<\gamma_2<\sqrt{P}$, are
realizable.

To see that, we take a function $f$ such that $||f||=1$,
$\gamma_1=\gamma(f,a)$, $\zeta=\zeta(f,b)$ where $c=ab$;

The spreading in time of the translation of $f$,

$$\gamma(f,a,x_0)=\int_{-\infty}^{\infty}|f(x-x_0)|^2{w_1}_a(x)dx$$

is a continuous function of $x_0$. $\gamma(f,a,0)=\gamma_1$ and

$$lim_{x_0\rightarrow\infty}\gamma(f,a,x_0)=\sqrt{P}$$

So for some $x_0$ $\gamma(f,a,x_0)=\gamma_2$.

The spreading in frequency of the translation of $f$ is a
constant,

$$\zeta=\frac{\big(\int_{-\infty}^{\infty}w_2(\omega)|\hat
f(\omega)|^2\big)^{\frac{1}{2}}}{||f||}$$

and therefore $(c,\gamma_2,\zeta)$ is realizable.

In the same way: If a point $(c,\gamma,\zeta_1)$ is realizable
then all points $(c,\gamma,\zeta_2)$, $\zeta_1<\zeta_2<\sqrt{Q}$,
are realizable.

Step 2:

$\forall (\gamma,\zeta)$ in the square $N$, $\exists c>0$ such
that $(c,\gamma,\zeta)$ is realizable:

To see that, we take an arbitrary $f$ with norm $||f||=1$. The
spreading functions $\gamma(f,a)$, $\zeta(f,b)$ are defined and
continuous on the interval $(0,\infty)$ as functions of $a$ and
$b$ respectively. $lim_{a \rightarrow \infty}\gamma(f,a)=0$,
$lim_{b \rightarrow \infty}\zeta(f,b)=0$ and therefore $\exists$
$\gamma_1<\gamma$, $\zeta_1<\zeta$, $c=ab$ such that the point
$(c,\gamma_1,\zeta_1)$ is realizable and by the first part of the
proof $(c,\gamma,\zeta)$ is realizable.

\begin{defn}

$\forall(\gamma,\zeta)\in N$ the function $\psi(\gamma,\zeta)$ is
defined as:

\begin{equation}
\label{definition_of_psi_fuction} \psi(\gamma,\zeta)=\inf\{c\mid
\mbox{ The point } (c,\gamma,\zeta) \mbox{is realizable}\}
\end{equation}

\end{defn}

Note that the points of $\partial N \times [0, \infty)$ are not
realizable, because of the properties of the weights.

Step 3:

The function $\psi_{w_1(x),w_2(\omega)}(\gamma,\zeta)$ is a
non-increasing function of $\gamma$ for fixed $\zeta$ since from
what we have shown above we get that for $\gamma_1<\gamma_2$ we
have the following inclusion

\begin{equation}
\{c\mid \mbox{The point }(c,\gamma_1,\zeta) \mbox{ is
realizable}\}\subseteq \{c\mid \mbox{The point }(c,\gamma_2,\zeta)
\mbox{ is realizable}\}
\end{equation}

In the same way we get that $\psi(\gamma,\zeta)$ is a non
-increasing function of $\zeta$ for fixed $\gamma$.

The case $w_1(x)=w_2(x)$ : If a point $(c,\gamma_1,\zeta_1)$ is
realizable, then there is a function such that $||f||=1$,
$\gamma(f,a)=\gamma_1$, $\zeta(f,b)=\zeta_1$ and $c=ab$. Since
$\hat{\hat{f}}(x)=f(-x)$, we have for $\hat{f}$, $||\hat{f}||=1$,
$\gamma(\hat{f},b)=\zeta_1$, $\zeta(\hat{f},a)=\gamma_1$ and
$c=ba$ which means that the point $(c,\zeta_1,\gamma_1)$ is
realizable and therefore $\psi(\gamma,\zeta)$ is symmetric.

Step 4:

If a point $(c_1,\gamma_0,\zeta_0)$ is realizable then the points
$(c_2,\gamma_0,\zeta_0)$, $c_2\geq c_1$ are realizable: Without
loss of generality we can take $a_1,b_1$ s.t. $a_1b_1=c_1$ and
$b_2$ s.t. $a_1b_2=c_2$. If $\exists f$ s.t.
$\gamma(f,a_1)=\gamma_0$, $\zeta(f,b_1)=\zeta_0$ then for this $f$
$\zeta(f,b_2)=\zeta_1<\zeta_0$ since

$$\big(\int_{-\infty}^{\infty}{w_2}_{b_1}(\omega)|\hat
f(\omega)|^2\big)^{\frac{1}{2}} ||f||^{-1}>
\big(\int_{-\infty}^{\infty}{w_2}_{b_2}(\omega)|\hat
f(\omega)|^2\big)^{\frac{1}{2}}||f||^{-1}$$

 from the monotonicity of $w_2(\omega)$.

This means that the point $(\gamma_0,\zeta_1,c_2)$ is realizable
and from step 1 it follows that $(\gamma_0,\zeta_0,c_2)$ is
realizable.

From our construction the set of points that fulfill inequality
\eqref{uncertainty_finite_lim} and the possibility body are equal
up to the graph of $\psi_{w_1(x),w_2(\omega)}(\gamma,\zeta)$,
which is a set of measure zero.

\end{proof}

\vspace{3 mm}

We will prove now a similar theorem for weights $w(x)$ s.t.
$\lim_{x\rightarrow \infty}w(x)=\infty$. The structure of the
proof is the same. We will state the theorem and explain the
necessary modifications in the proof.

We will define the type of weights we will have in the theorem.
The weight in the HPW case, $w(x)=x^2$, is an example of a weight
of this type. Property 1 of the weights will use the function
$C(h,x)$ which is related to the weight as follows:

\begin{defn}
$$ C(h,x_*)=\sup\{\frac{w(x+h)}{w(x)}\mid x\geq x_*\}.$$
$$\forall (h,x_*)\in (0,\infty)\times (0,\infty)$$
\end{defn}

Note that $C(h,x)$ is a non-increasing function of $x$ for fixed
$h$. Note that if $w(x)$ is a non-decreasing function on $[0,
\infty)$ then $C(h,x)$ is a non-decreasing function of $h$ for
fixed $x$ since $h_2\geq h_1>0$ implies $\frac
{w(x+h_2)}{w(x)}\geq \frac {w(x+h_1)}{w(x)}$.

\vspace{3 mm}

\begin{defn}[Property 1]

A weight has property 1 iff, $\forall h_0>0$ $\exists x_0>0$ such
that $C(h_0,x_0)$ is finite.

\end{defn}
\vspace{3 mm}

\begin{defn} [Weight of type $\infty$]

A weight $w(x)$ is of type $\infty$ if and only if:

\begin{itemize}

\item[a)] $w(0)=0$.

\item[b)] $w(x)$ is a continuous function.

\item[c)] $w(x)$ is an even function i.e. $w(x)=w(-x)$.

\item[d)] $w(x)$ is strictly increasing on $[0,\infty]$ i.e. $\forall
0\leq x<y$ $w(x)<w(y)$.

\item[e)] $w(x)$ tends to $\infty$ as $x$ tends to $\infty$ not faster than some polynomial, i.e.

$\lim_{x\rightarrow \infty}w(x) = \infty$, $w_1(x)\leq P(x)$ where
$P(x)$ is a polynomial.

\item[f)] $w(x)$ obtains property 1.
\end{itemize}

\end{defn}

\vspace{3 mm}

\begin {thm}[General Uncertainty Theorem For Weights Of Type
$\infty$]

 Let $w_1(x)$ and $w_2(x)$ be weights of type $\infty$

Then the possibility body,$PB_{\gamma^1,\zeta^1}$, is defined,up
to a set of measure zero, by a generalized uncertainty inequality
of the form:

\begin{equation}
ab=c\geq \psi_{w_1(x),w_2(\omega)}(\gamma,\zeta)
\end{equation}

where $\psi$ is defined on the open upper-right quarter of the
plain  $N=(0,\infty)\times (0,\infty)$.

$\psi(\gamma,\zeta)$ is a non-increasing function of $\gamma$ for
fixed $\zeta$ and a non-increasing function of $\zeta$ for fixed
$\gamma$.

If $w_1(x)=w_2(x)$ then $\psi_{w_1(x),w_2(\omega)}(\gamma,\zeta)$
is symmetric.

\label{generalized_uncertainty_theorem2}
\end{thm}

%\prf :
\begin{proof}
In Theorem \ref{generalized_uncertainty_theorem1}, in step 1, the
continuity of

$\gamma(f,a,x_0)=\int_{-\infty}^{\infty}|f(x-x_0)|^2{w_1}_a(x)dx$
in $x_0$ was obvious. Here Some elaboration is needed. We will
proof continuity from the right (continuity from the left can be
done in the same way). We will live it as an exercise to show that
if $\gamma(f,a,x_0)$ exists then $\gamma(f,a,x_1)$ exists for all
$x_1>x_0$. Note that it is enough to show continuity at
$\gamma(f,a,0)$.

So we show first continuity at $x_0=0$ of $\gamma(f,a,x_0)$. We
fix an arbitrary positive number $h_0>0$, and choose $q_1<0$ such
that $\int_{-\infty}^{q_1}f(x)w(x)dx<\epsilon_1$. From property 1,
$\exists$ $x_p>h_0$ s.t. $C(h_0,x_p)$ is finite. We choose
$q_2>x_p$ such that $\int_{q_2}^\infty |f(x)|^2 w(x)dx<\epsilon_3$

We have:

\begin{align}
\int_{-\infty}^{\infty}(|f(x-h)|^2-|f(x)|^2)w(x)dx&=\int_{-\infty}^{q_1}(|f(x-h)|^2-|f(x)|^2)w(x)dx+ \notag \\
&+\int_{q_1}^{q_2+h_0}(|f(x-h)|^2-|f(x)|^2)w(x)dx\notag \\
&+\int_{q_2+h_0}^{\infty}(|f(x-h)|^2-|f(x)|^2)w(x)dx
\label{tree_terms}
\end{align}

We will check the three terms separately:

The first term of \eqref{tree_terms}:

$$\forall 0<h<|q_1|$$

\begin{equation}
\int_{-\infty}^{q_1}|f(x-h)|^2w(x)dx=\int_{-\infty}^{q_1-h}|f(y)|^2w(y+h)dy\leq
\int_{-\infty}^{q_1}|f(y)|^2w(y+h)dy\leq
\end{equation}

$$\leq \int_{-\infty}^{q_1}|f(y)|^2w(y)dy<\epsilon_1$$

and from the triangle inequality we have:

\begin{equation*}
\int_{-\infty}^{q_1}(|f(x-h)|^2-|f(x)|^2)w(x)dx<\int_{-\infty}^{q_1}(||f(x-h)|^2-|f(x)|^2|)w(x)dx<2\epsilon_1
\end{equation*}

The second term of \eqref{tree_terms}:

\begin{defn}
 $D=\max\{w(q_1),w(q_2+h_0)\}.$
\end{defn}

$\forall \epsilon_2>0, \exists 0<h_1<\min\{|q_1|,h_0\}$ s.t.

\begin{equation*}
\int_{q_1}^{q_2+h_0}(|f(x-h_1)|^2-|f(x)|^2)w(x)dx<D\int_{-\infty}^{\infty}(||f(x-h_1)|^2-|f(x)|^2|)dx<\epsilon_2
\end{equation*}

since $|f(x)|^2\in L_1$ and $\forall g(x)\in L_1$ we have
$\lim_{h\rightarrow 0}\int_{-\infty}^{\infty}|g(x-h)-g(x)|=0$

The third term of \eqref{tree_terms}:

\begin{align*}
\int_{q_2+h_0}^{\infty}|f(x-h_1)|^2w(x)dx=& \int_{q_2+h_0-h_1}^{\infty}|f(y)|^2w(y)\frac{w(y+h_1)}{w(y)}dy< \\
& <C(h_0,x_0)\epsilon_3
\end{align*}

where we used the monotonicity of $C(h,x)$. By the triangle
inequality we have:

\begin{align}
&\int_{q_2+h_0}^{\infty}(|f(x-h_1)|^2-|f(x)|^2)w(x)dx\leq \notag \\
&\leq \int_{q_2+h_0}^{\infty}(||f(x-h_1)|^2-|f(x)||^2)w(x)dx\leq \notag \\
& \leq (1+C(h_0,x_0))\epsilon_3 \label{tree_terms}
\end{align}

Taking $\epsilon_1=\frac{\epsilon}{3}$,
$\epsilon_2=\frac{\epsilon}{3}$ and
$\epsilon_3=\frac{1}{3(1+C(h_0,x_p))}$ and $h_1$ as above we get
the continuity of $\gamma(f,a,x_0)$ as a function of $x_0$ at
$x_0=0$.

In step 2 the only modification we need is that instead of taking
an arbitrary $f(x)$ we take $f(x)=exp(-x^2)$.

No modification in steps tree and four is needed.

\end{proof}

\vspace{3 mm}

Note that in Theorem \ref{generalized_uncertainty_theorem1}
$N=(0,\sqrt{P})\times (0,\sqrt{Q})$ and in Theorem
\ref{generalized_uncertainty_theorem2} $N=(0, \infty) \times (0,
\infty)$.

Now we define similar definitions as in the LPS case (Definitions
\ref{LPpossible_area} and \ref{LPimpossible_area}) for weights of
type 0 and weights of type $\infty$:

\begin{defn} The possibility area of level c is the set:
$M_c=\{(p,q)| (p,q,c)$ is realizable $\}$
\end{defn}

\begin{defn} The impossible area of level c is the set
$M'_c=N\setminus M_c$
\end{defn}

Note that we have already defined the possibility body at
Definition \ref{possibility_body}.

 To the pair $(M_c,M'_c)$ we will call "The possibility map
of level c".

\subsection{The Heisenberg-Pauli-Weyl General Uncertainty Principle}
\label{subsec_Heisenberg-Pauli-Weyl}
 It is easy to see that the HPW weights i.e.
$w_1(x)=w_2(x)=x^2$ are of type $\infty$. We use the HPW
uncertainty principle to compute the function $\psi(\gamma,\zeta)$
explicitly and then we find the boundary of the possible area of
level c (see
Figure~\ref{Heisenberg_Pauli_Weyl_Spreading_possibility_maps_zeta_gamma_fig}).

 \begin{figure}[!h]
  \centerline{
    \mbox{\includegraphics[scale=0.4]{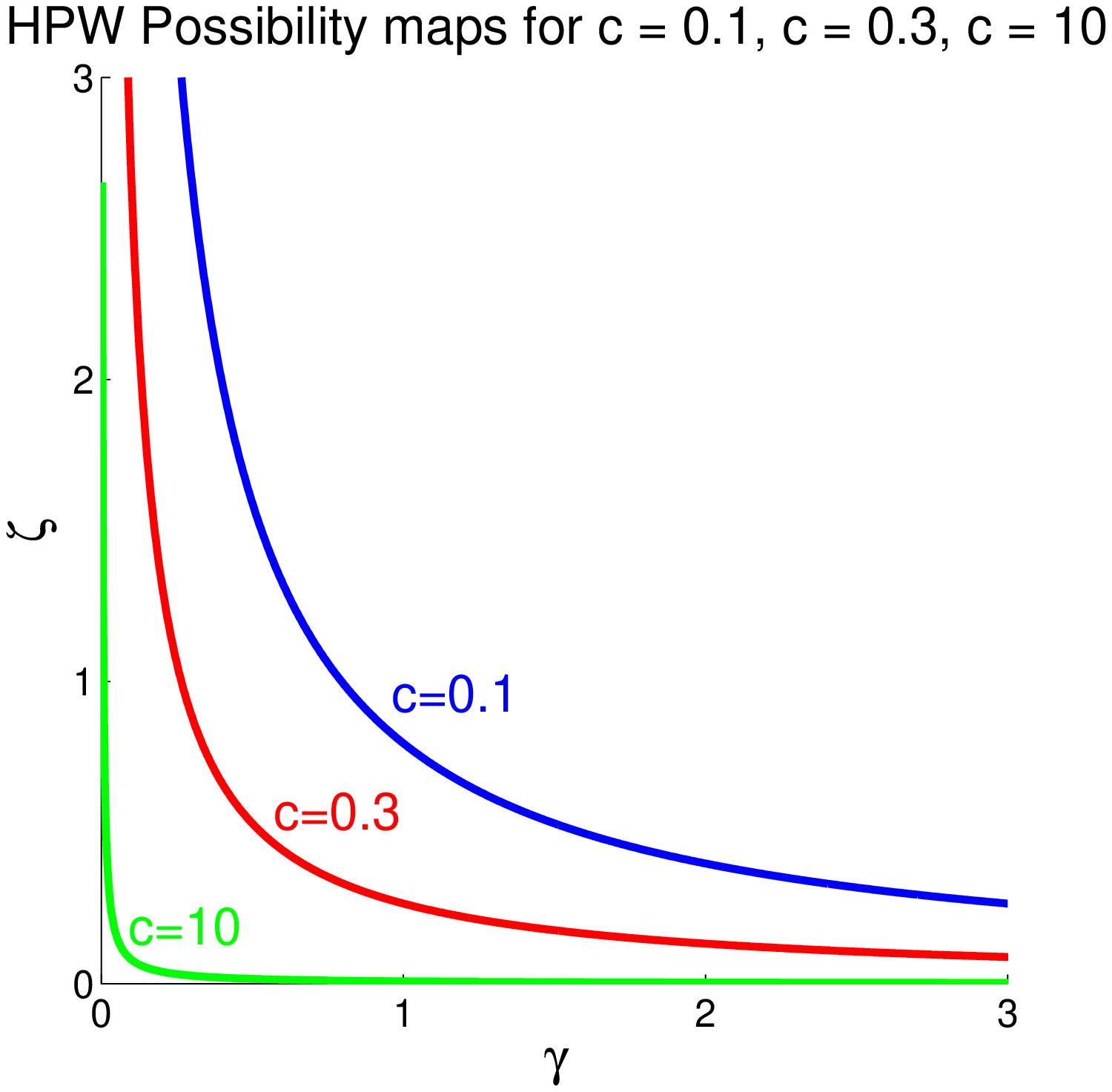}}
    \mbox{\includegraphics[scale=0.4]{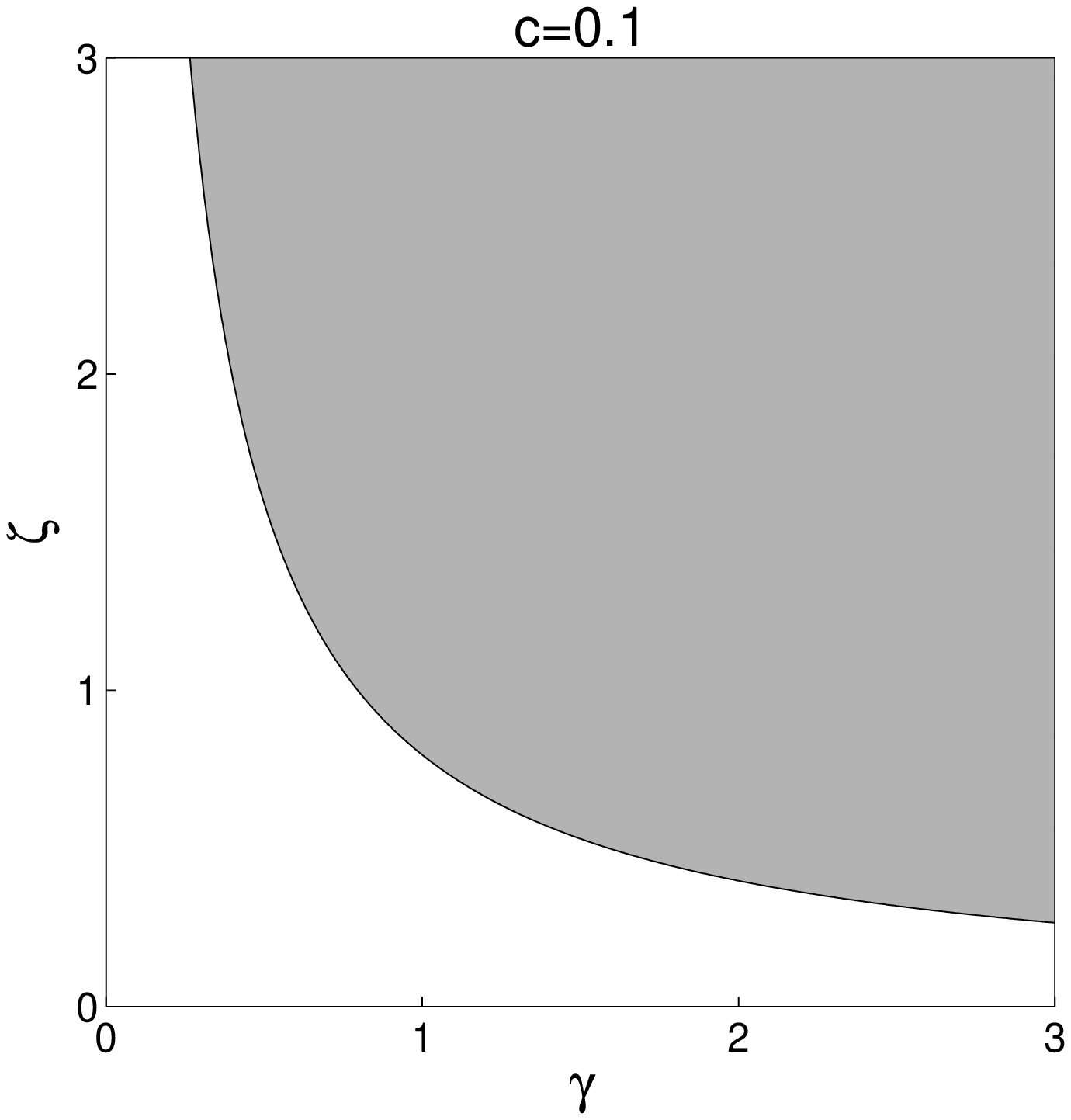}}
  }
  \centerline{
    \mbox{\includegraphics[scale=0.4]{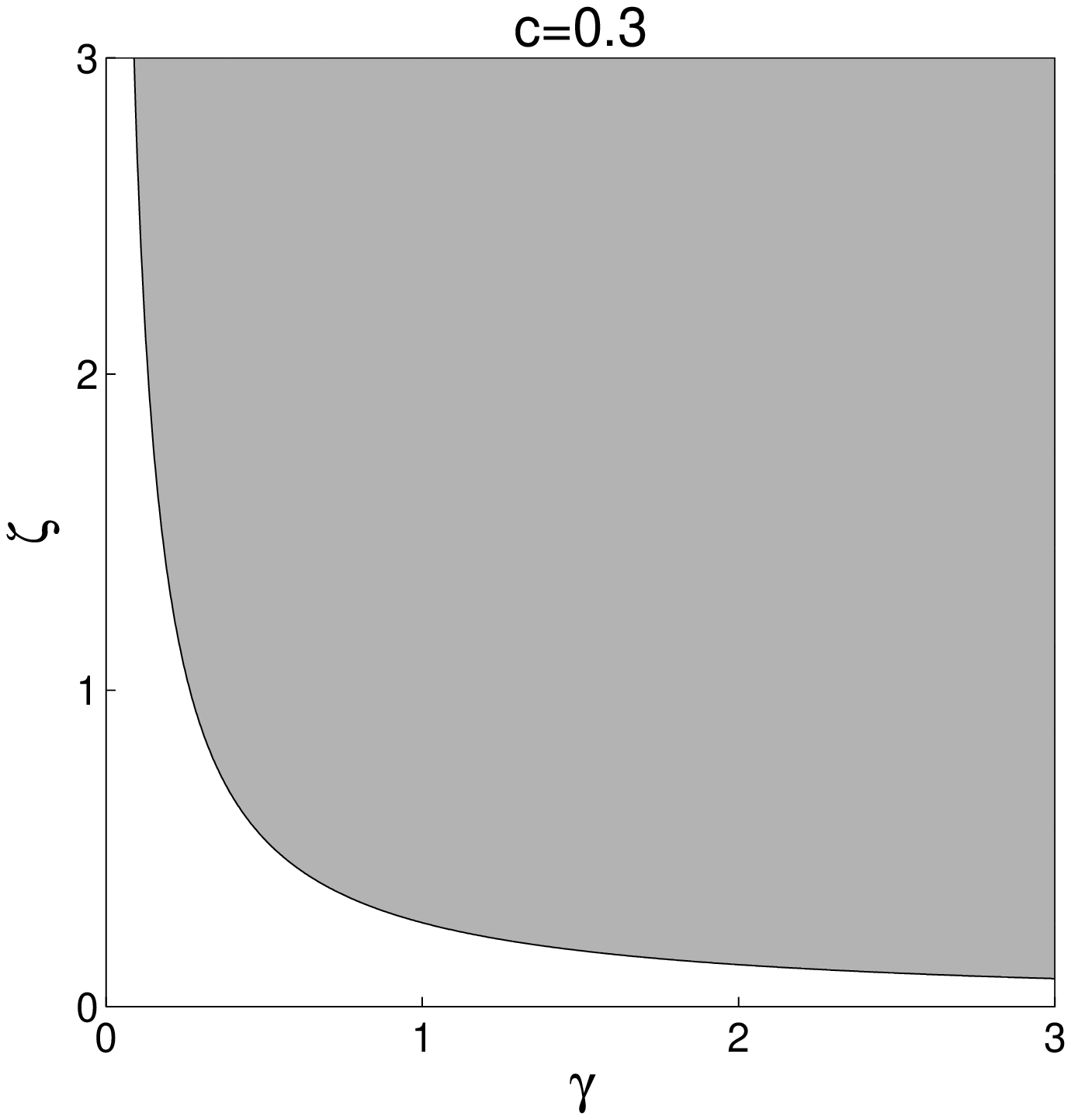}}
    \mbox{\includegraphics[scale=0.4]{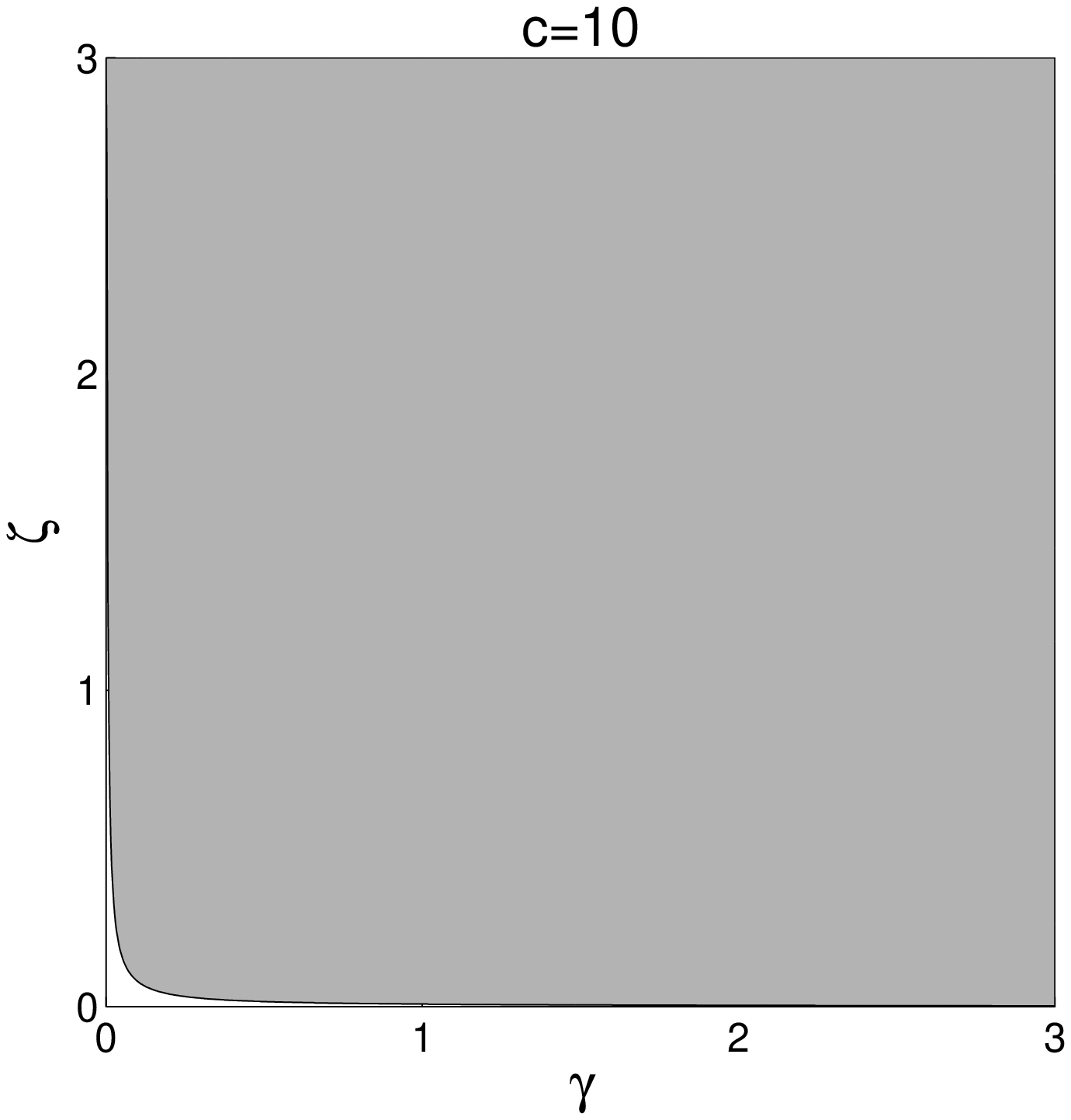}}
  }
  \caption{ Heisenberg-Pauli-Weyl Possibility maps for $c=0.1$, $c=0.3$,
$c=10$ using the coordinates (spreading parameters) $\zeta$ and
$\gamma$.}
  \label{Heisenberg_Pauli_Weyl_Spreading_possibility_maps_zeta_gamma_fig}
  \end{figure}

from:
\begin{equation}
\label{Heisebnerg_multiplication_of_spreadings}
\int_{-\infty}^\infty (\frac{x}{a})^2|f(x)|^2 dx
\int_{-\infty}^\infty (\frac{\omega}{b})^2|\hat f(\omega)|^2
d\omega = (\frac{1}{a b})^2 \int_{-\infty}^\infty x^2|f(x)|^2 dx
\int_{-\infty}^\infty \omega^2|\hat f(\omega)|^2 d\omega
\end{equation}

we write

\begin{equation*}
a^2\gamma^2(f,a)b^2\zeta^2(f,b)=\gamma^2(f,1)\zeta^2(f,1)\geq(\frac{1}{4\pi})^2
\end{equation*}

or

\begin{equation*}
a^2b^2\geq(\frac{1}{4\pi})^2\frac{1}{\gamma^2(f,a)\zeta^2(f,b)}
\end{equation*}

we define:

\begin{equation}
\label{HPW_theta_function}
\theta(\gamma,\zeta)=\frac{1}{4\pi}\frac{1}{\gamma\zeta}
\end{equation}

and get:

\begin{equation*}
c=ab \geq \theta(\gamma,\zeta)
\end{equation*}

We define:

\begin{equation}
u_d(x)=\sqrt{\frac{d}{\sqrt{2}}}e^{-\pi(\frac{x}{d})^2}
\end{equation}

We live it as an exercise to check that $||u_d(x)||=1$,
$\gamma(u_d(x),a)$ is continuous as a function of $d$,
$\lim_{d\rightarrow 0} \gamma(u_d(x),a)=0$ and $\lim_{d\rightarrow
\infty} \gamma(u_d(x),a)=\infty$.

from HPW Theorem and
\eqref{Heisebnerg_multiplication_of_spreadings} it follows that
$\forall a,b,d
>0$:

\begin{equation}
\gamma(u_d(x),a) \zeta(u_d(x),b)ba=\frac{1}{4\pi}
\end{equation}

and therefore the graph of $\theta(\gamma,\zeta)$ is equal:

\begin{equation}
\bigcup _{d\in (0,\infty),c \in (0,\infty)} (\gamma(u_d(x),a),
\zeta(u_d(x),b),c)
\end{equation}

(where again the dependence on a and b separately is not
important), which means that the infimum in
\eqref{definition_of_psi_fuction} is attained and

\begin{equation}
\label{HPW_psi_fumction}
 \psi(\gamma,\zeta)=\theta(\gamma,\zeta)
\end{equation}

$\Box$
%$\diamondsuit$

\vspace{3 mm}

From \eqref{HPW_theta_function} and \eqref{HPW_psi_fumction} we
see that the boundary of the possible area of level c in the HPW
case, when we use the spreading parameters is (see
Figure~\ref{Heisenberg_Pauli_Weyl_Spreading_possibility_maps_zeta_gamma_fig}):

\begin{equation}
\gamma=\frac{1}{4\pi c \zeta}
\end{equation}

\vspace{7 mm}

\section{Natural coordinates systems for the possibility body and
convexity} \label{sec5}

 As we saw the notion of generalized
uncertainty principles has different settings. In LP Theorem
(Theorem \ref{LPtheorem}) we use concentration parameters and the
function $\phi(\alpha,\beta)$ is defined on $D_{LP}$. In LP*
Theorem (Theorem \ref{LP2theorem}) we use spreading parameters and
the function $\psi(\gamma,\zeta)$ is defined on $D_{LP^*}$. In
General Uncertainty Theorem For Weights Of Type 0 (Theorem
\ref{generalized_uncertainty_theorem1}) and General Uncertainty
Theorem For Weights Of Type $\infty$ (Theorem
\ref{generalized_uncertainty_theorem2}) we use spreading
parameters and $\psi(\gamma,\zeta)$ is defined on
$(0,\sqrt{P})\times (0,\sqrt{Q})$ and $(0,\infty)\times
(0,\infty)$ respectively. The different settings of the "General
uncertainty principles" indicates that $\textit{there is a new
phenomenon underline those}$.

We saw that we can use different coordinates for representing the
generalized uncertainty principles (see \ref{uncertainty2} and
\ref{uncertainty3}). Below we see that it is equivalent to
measuring concentration (spreading) in different ways. In this
section we discuss the question of the existence of natural
coordinates for describing the phenomenon of "generalized
uncertainty principles" and the question of the convexity of the
possibility body.

We start by showing that the boundary of the possible area in the
LP case is an algebraic curve, when we use the concentration
parameters $\alpha$ and $\beta$ (see \eqref{Time_concentration},
\eqref{Frequency_concentration} and LP Theorem - Theorem
\ref{LPtheorem}) or the spreading parameters $\gamma$ and $\zeta$
(see definitions \ref{Time_spreading} and
\ref{Frequency_spreading} and LP* Theorem - Theorem
\ref{LP2theorem})

Using the concentration parameters we have:

\begin{align*}
&\cos^{-1}\alpha + \cos^{-1}\beta                  =\cos^{-1}\sqrt{\lambda_0} \notag \\
&\alpha\beta -\sqrt{1-\alpha^2}\sqrt{1-\beta^2}           =\sqrt{\lambda_0}\notag \\
&\alpha^2\beta^2-2\sqrt{\lambda_0}\alpha\beta+\lambda_0   =(1-\alpha^2)(1-\beta^2) \notag\\
&\alpha^2+\beta^2-2\sqrt{\lambda_0}\alpha\beta =1-\lambda_0 \notag\\
&\frac{1}{1-\lambda_0}\alpha^2+\frac{1}{1-\lambda_0}\beta^2-\frac{\sqrt{\lambda_0}}{1-\lambda_0}\alpha\beta
=1 \label{tree_terms}
\end{align*}

\begin{equation*}
\begin{pmatrix}
\alpha & \beta
\end{pmatrix}
\frac{1}{1-\lambda_0}
\begin{pmatrix}
1 & -\sqrt{\lambda_0} \\
-\sqrt{\lambda_0} &  1
\end{pmatrix}
\begin{pmatrix}
\alpha \\
\beta
\end{pmatrix}
=1
\end{equation*}

\begin{equation}
\begin{pmatrix}
\alpha & \beta
\end{pmatrix}
\begin{pmatrix}
\frac{1}{\sqrt{2}} & \frac{1}{\sqrt{2}} \\
\frac{1}{\sqrt{2}} &  -\frac{1}{\sqrt{2}}
\end{pmatrix}
\frac{1}{1-\lambda_0}
\begin{pmatrix}
1-\sqrt{\lambda_0} & 0 \\
0  & 1+\sqrt{\lambda_0}
\end{pmatrix}
\begin{pmatrix}
\frac{1}{\sqrt{2}} & \frac{1}{\sqrt{2}} \\
\frac{1}{\sqrt{2}} &  -\frac{1}{\sqrt{2}}
\end{pmatrix}
\begin{pmatrix}
\alpha \\
\beta
\end{pmatrix}
=1
\end{equation}

and we see that using the coordinates:

\begin{equation*}
\begin{pmatrix}
u \\
v
\end{pmatrix}
=
\begin{pmatrix}
\frac{1}{\sqrt{2}} & \frac{1}{\sqrt{2}} \\
\frac{1}{\sqrt{2}} &  -\frac{1}{\sqrt{2}}
\end{pmatrix}
\begin{pmatrix}
\alpha \\
\beta
\end{pmatrix}
\end{equation*}

we have the ellipse in simple form:

\begin{equation}
\frac{u}{\big( \sqrt{\frac{1-\lambda_0}{1-\sqrt{\lambda_0}}}\big)
^2} + \frac{v}{\big(
\sqrt{\frac{1-\lambda_0}{1+\sqrt{\lambda_0}}}\big) ^2}=1
\end{equation}

Thus using the concentration parameters $\alpha$ and $\beta$ the
boundary of the possibility area consists of straight lines and
part of an ellipse (see Figure~\ref{Ellipses_fig}) which its main
axis is in the direction

\begin{equation*}
\begin{pmatrix}
 \frac{1}{\sqrt{2}} \\
 \frac{1}{\sqrt{2}}
\end{pmatrix}
\end{equation*}

and its minor axis is in the direction

\begin{equation*}
\begin{pmatrix}
 \frac{1}{\sqrt{2}} \\
 - \frac{1}{\sqrt{2}}
\end{pmatrix}
\end{equation*}

We calculate the distance of the ellipse focuses from the origin

\begin{align}
c=\Big( \Big( \sqrt{\frac{1-\lambda_0}{1-\sqrt{\lambda_0}}}\Big)
^2 - \Big( \sqrt{\frac{1-\lambda_0}{1+\sqrt{\lambda_0}}}\Big) ^2
\Big)^\frac{1}{2} \notag \\
= 2^\frac{1}{2} \lambda_0^\frac{1}{4}
\end{align}

and see that the focuses of the ellipse are placed at

\begin{equation*}
2^\frac{1}{2} \lambda_0^\frac{1}{4}
\begin{pmatrix}
 \frac{1}{\sqrt{2}} \\
 \frac{1}{\sqrt{2}}
\end{pmatrix}
\end{equation*}

and

\begin{equation*}
-2^\frac{1}{2} \lambda_0^\frac{1}{4}
\begin{pmatrix}
 \frac{1}{\sqrt{2}} \\
 \frac{1}{\sqrt{2}}
\end{pmatrix}
\end{equation*}

 \begin{figure}
  \centerline{
    \mbox{\includegraphics[scale=0.4]{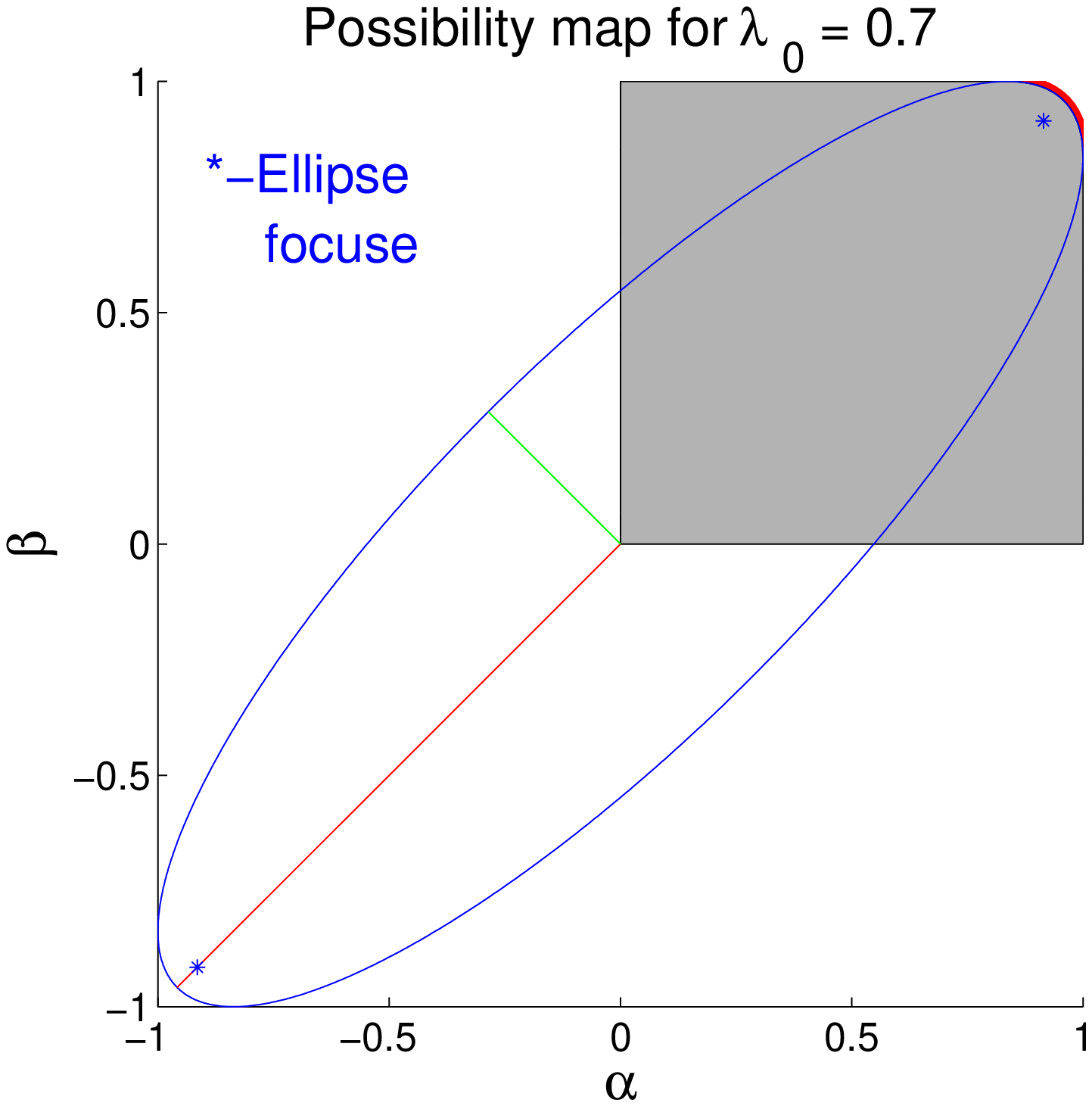}}
    \mbox{\includegraphics[scale=0.4]{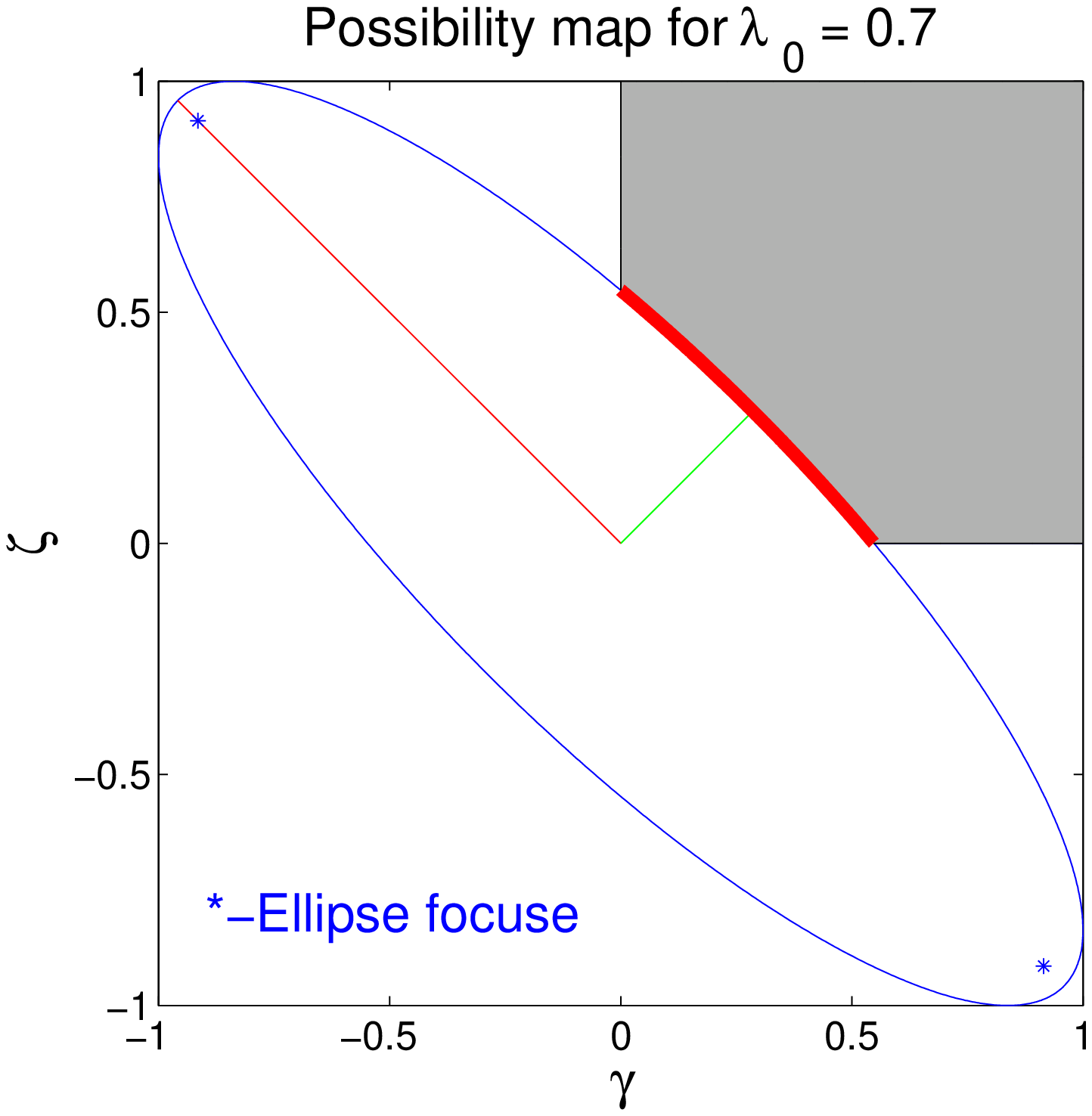}}
  }
  \caption{ Possibility maps and the related ellipses for $\lambda_0=0.7$, using the coordinates (concentration parameters)
$\alpha$ and $\beta$ on the left and using the coordinates
(spreading parameters) $\zeta$ and $\gamma$ on the right.}
  \label{Ellipses_fig}
  \end{figure}

Now we show that the boundary of the possibility area consists of
straight lines and part of an ellipse (see
Figure~\ref{Ellipses_fig}) when we use the spreading parameters
$\gamma$ and $\zeta$.

Similar calculations to the concentration case above gives:

\begin{equation*}
\cos^{-1}\sqrt{1-\gamma^2} +
\cos^{-1}\sqrt{1-\zeta^2}=\cos^{-1}\sqrt{\lambda_0}
\end{equation*}

\begin{equation*}
\begin{pmatrix}
\gamma & \zeta
\end{pmatrix}
\frac{1}{1-\lambda_0}
\begin{pmatrix}
1 & \sqrt{\lambda_0} \\
\sqrt{\lambda_0} &  1
\end{pmatrix}
\begin{pmatrix}
\gamma \\
\zeta
\end{pmatrix}
=1
\end{equation*}

\begin{equation*}
\begin{pmatrix}
\gamma & \zeta
\end{pmatrix}
\begin{pmatrix}
\frac{1}{\sqrt{2}} & \frac{1}{\sqrt{2}} \\
\frac{1}{\sqrt{2}} &  -\frac{1}{\sqrt{2}}
\end{pmatrix}
\frac{1}{1-\lambda_0}
\begin{pmatrix}
1+\sqrt{\lambda_0} & 0 \\
0  & 1-\sqrt{\lambda_0}
\end{pmatrix}
\begin{pmatrix}
\frac{1}{\sqrt{2}} & \frac{1}{\sqrt{2}} \\
\frac{1}{\sqrt{2}} &  -\frac{1}{\sqrt{2}}
\end{pmatrix}
\begin{pmatrix}
\gamma \\
\zeta
\end{pmatrix}
=1
\end{equation*}

In this case the main axis of the ellipse is in the direction

\begin{equation*}
\begin{pmatrix}
 \frac{1}{\sqrt{2}} \\
 - \frac{1}{\sqrt{2}}
\end{pmatrix}
\end{equation*}

and its minor axis is in the direction

\begin{equation*}
\begin{pmatrix}
 \frac{1}{\sqrt{2}} \\
 \frac{1}{\sqrt{2}}
\end{pmatrix}
\end{equation*}

The focuses of the ellipse are at the points

\begin{equation*}
2^\frac{1}{2} \lambda_0^\frac{1}{4}
\begin{pmatrix}
 \frac{1}{\sqrt{2}} \\
 -\frac{1}{\sqrt{2}}
\end{pmatrix}
\end{equation*}

and

\begin{equation*}
-2^\frac{1}{2} \lambda_0^\frac{1}{4}
\begin{pmatrix}
 \frac{1}{\sqrt{2}} \\
 -\frac{1}{\sqrt{2}}
\end{pmatrix}
\end{equation*}

In the HPW case (see \eqref{HPW_theta_function} and
\eqref{HPW_psi_fumction}), when we use $\gamma^m$, $m\in N$ and
$\zeta^n$, $n\in N$ as parameters (we use the notation
$\Psi(\gamma^m,\zeta^n)$) we get:

\begin{equation}
\label{Psi_gamma_m_zeta_n}
\Psi(\gamma^m,\zeta^n)=\psi(\gamma,\zeta)
\end{equation}

i.e

\begin{equation}
\Psi(\gamma^m,\zeta^n)=\frac{1}{4\pi}\frac{1}{\sqrt[m]{\gamma^m}\sqrt[n]{\zeta^n}}
\end{equation}

We see that the boundary of the possibility area of level c is
again an algebraic curve defined by:

\begin{equation}
(\gamma^m)^n (\zeta^n)^m=\frac{1}{(4\pi c)^{m n}}
\end{equation}

The question of finding types of weights and coordinate systems
such that the boundary of the possible areas of different levels
(different c's) are algebraic curves may indicate what are the
natural coordinate systems for describing the general uncertainty
principles phenomenon.

As we saw in section \ref{sec3} the question of convexity of the
possibility body depends on the set of parameters we choose for
describing the possibility body. (In section \ref{sec3} we had a
convex possible area for the parameter set $(\gamma^2,\zeta^2,c)$
and a non-convex possible area for the parameter set
$(\gamma,\zeta,c)$.) The following definitions (Definitions
\ref{realizable_point4_of_oredr_m},
\ref{realizable_point3_of_order_m} and
\ref{possibility_body_of_order_m}) relates the coordinate systems
to different ways of measuring the time and frequency spreadings.
Motivated by the HPW case we generalize definitions
\ref{realizable_point4}, \ref{realizable_point3} and
\ref{possibility_body} as follows:

\begin{defn}
\label{realizable_point4_of_oredr_m}
 A point $(p,q,a,b)\in R^{4+}$ is called realizable with respect to $\gamma^m$, and
 $\zeta^m$, $m\in N$ iff there is a function $f\in L_2$ such that:
 $p=\gamma^m(f,a)$ and
$q=\zeta^m(f,b),a>0,b>0$
\end{defn}

Note that one can think about the symbol $\gamma^m$ (respectively
$\zeta^m$) as the function $\gamma$ (respectively $\zeta$) to the
power $m$ or as a new function for measuring spreading.
In~\ref{Psi_gamma_m_zeta_n} we use the symbols $\gamma^m$ and
$\zeta^m$ also as a spreading coordinates. Bellow we will continue
to use $\gamma^m$ and $\zeta^n$ with their different meanings. The
meaning will be clear from the context. Using Lemma
\ref{realizable_point_lemma} and the note that follows it we can
define:

\begin{defn}
\label{realizable_point3_of_order_m}
 A point $(p,q,c)\in R^{3+}$ is called realizable with respect to $\gamma^m$ and
 $\zeta^m$, $m\in N$ iff there is a function $f\in L_2$ such that, $||f||=1$, $p=\gamma^m(f,a)$,
$q=\zeta^m(f,b)$ and $c=a b$.
\end{defn}

\begin{defn}
\label{possibility_body_of_order_m} We will call the set of
realizable points with respect to $\gamma^m$, and
 $\zeta^m$, $m\in N$, in $R^{3+}$ "possibility body of order m" and
denote it by $PB_{\gamma^m,\zeta^m}$.
\end{defn}

It is easy to see that the following relation holds:

\begin{equation}
\label{possibility_bodies_relation} (p^m,q^m,c)\in
PB_{\gamma^m,\zeta^m} \text{ iff } (p,q,c)\in
PB_{\gamma^1,\zeta^1}
\end{equation}

In the following we will continue to discuss the case of weights
of type $\infty$.

\begin{lem}
\label{argument_1} For every fixed $\gamma_1>0$ and $c_1>0$, there
exists a point $(\gamma_1,\zeta,c_1)$ that is realizable.

For every fixed $\zeta_1$ and $c_1$, there exists a point
$(\gamma,\zeta_1,c_1)$ that is realizable.
\end{lem}

%\prf :
\begin{proof}

We fix $a,b>0$ such that $ab=c_1$

From the fact that

$$lim_{d\rightarrow 0}\gamma(u_d(x),a)=0$$

and the properties of our weights it follows that $\exists d>0$
s.t. $\gamma(u_d(x),a)<\gamma_1$ and

$$(\gamma(u_d(x),a),\zeta(u_d(x),b),c)$$

is realizable, and from step 1 of Theorem
\ref{generalized_uncertainty_theorem1} and its modification at
Theorem \ref{generalized_uncertainty_theorem2} it follows that the
point $(\gamma_1,\zeta,c_1)$ is realizable.

The existence of a realizable point $(\gamma,\zeta_1,c_1)$ for
every fixed $\zeta_1>0$ and $c_1>0$ is proved in the same way.

\end{proof}

Now we can define:

\begin{defn}
\label{psi_with_different_indices}
 We define
$\psi_{w_1(x),w_2(\omega)}^{\gamma^m}(\zeta^m,c)$,
$\psi_{w_1(x),w_2(\omega)}^{\zeta^m}(\gamma^m,c)$ and

$\psi_{w_1(x),w_2(\omega)}^m(\gamma^m,\zeta^m)$ on
$D=(0,\infty)\times (0,\infty)$ by:

\begin{equation}
\psi_{w_1(x),w_2(\omega)}^{\gamma^m}(\zeta^m,c)=inf\{\gamma^m \mid
(\gamma^m,\zeta^m,c)\in PB_{\gamma^m,\zeta^m}\}
\end{equation}

\begin{equation}
\psi_{w_1(x),w_2(\omega)}^{\zeta^m}(\gamma^m,c)=inf\{\zeta^m \mid
(\gamma^m,\zeta^m,c)\in PB_{\gamma^m,\zeta^m}\}
\end{equation}

\begin{equation}
\psi_{w_1(x),w_2(\omega)}^m(\gamma^m,\zeta^m)=inf\{c \mid
(\gamma^m,\zeta^m,c)\in PB_{\gamma^m,\zeta^m}\}
\end{equation}

\end{defn}

Basically we have just change our point of view concerning the
general uncertainty principles and the following facts are easy to
see:

a) The general uncertainty principles can be written also in the
forms:

\begin{equation}
\gamma^m \geq \psi_{w_1(x),w_2(\omega)}^{\gamma^m}(\zeta^m,c)
\end{equation}

and

\begin{equation}
\zeta^m\geq \psi_{w_1(x),w_2(\omega)}^{\zeta^m}(\gamma^m,c)
\end{equation}

where $\psi_{w_1(x),w_2(\omega)}^{\gamma^m}(\zeta^m,c)$ and
$\psi_{w_1(x),w_2(\omega)}^{\zeta^m}(\gamma^m,c)$ have the same
properties as $\psi_{w_1(x),w_2(\omega)}(\gamma,\zeta)$ (see
Theorem \ref{generalized_uncertainty_theorem2}).

b)
$\psi_{w_1(x),w_2(\omega)}^m(\gamma^m,\zeta^m)=\psi_{w_1(x),w_2(\omega)}(\gamma,\zeta)$.

c) If $w_1(x)=w_2(x)$ then
$\psi_{w_1(x),w_2(\omega)}^{\gamma^m}(\sigma,c)=\psi_{w_1(x),w_2(\omega)}^{\zeta^m}(\sigma,c)$.

Now we will focus on homogeneous weights. We will show that
$\forall m\in N$ the possibility body of order m is convex, and
find explicitly the related general uncertainty principles.

When we use homogenous weights of degree $k$ and the parameter set
$(\gamma^m,\zeta^m,c)$ we will use the following notation:

The possible area of level c will be denoted by: $M_c^{k,m}$

The possibility body will be denoted by: $PB_{\gamma^m,\zeta^m}^k$

\vspace{3 mm}

\begin{lem}
\label{homogeneity_lemma}
 If $w_2(\omega)$ is homogeneous of order $k$ (i.e.
$w_2(g\omega)=g^k w_2(\omega)$ then

\begin{equation}
(p,q)\in M_{1}^{k,m} \mbox{ iff }(p,c_0^{-\frac{km}{2}}q)\in
M_{c_0}^{k,m}
\end{equation}

and

\begin{equation}
\label{relation_in_psi_from_homogeneity}
\psi_{w_1(x),w_2(\omega)}^{\zeta^m}(\gamma^m,c_0)=c_0^{-\frac{km}{2}}\psi_{w_1(x),w_2(\omega)}^{\zeta^m}(\gamma^m,1)
\end{equation}

If $w_1(x)$ is homogeneous of order k (i.e. $w_1(gx)=g^k w_2(x)$
then

\begin{equation}
(p,q)\in M_{1}^{k,m} \mbox{ iff } (c_0^{-\frac{km}{2}}p,q)\in
M_{c_0}^{k,m}
\end{equation}

and

\begin{equation}
\psi_{w_1(x),w_2(\omega)}^{\gamma^m}(\zeta^m,c_0)=c_0^{-\frac{km}{2}}\psi_{w_1(x),w_2(\omega)}^{\gamma^m}(\zeta^m,1)
\end{equation}

\end{lem}

% \prf

\begin{proof}

Since
\begin{equation}
\zeta^m(f,c_0)=\frac{(\int_{-\infty}^\infty
{w_2}_{c_0}(\omega)|\hat{f}(\omega)|^2d\omega)^\frac{m}{2}}{||f||}=\frac{(\int_{-\infty}^\infty
\frac{1}{c_0^k}w_2(\omega)|\hat{f}(\omega)|^2d\omega)^\frac{m}{2}}{||f||}=
\end{equation}

$$=c_0^{-\frac{km}{2}}\zeta^m(f,1)$$

and since without loss of generality, we can take $a=1, b=c_0,
c_0=ab$, for calculating the map of level $c_0$ and we can take
$a=1, b=1, 1=ab$ for calculating the map of level 1, we have:

$$M_{c_0}^{k,m}=\bigcup_{f\in L^2}(\gamma^m(f,1),\zeta^m(f,c_0))=\bigcup_{f\in
L^2}(\gamma^m(f,1),c_0^{-\frac{km}{2}}\zeta^m(f,1))$$

we get that $(p,q)\in M_{1}^{k,m}$ iff
$(p,c_0^{-\frac{km}{2}}q)\in M_{c_0}^{k,m}$ and
\eqref{relation_in_psi_from_homogeneity} follows.

The second part of the lemma is done in the same way.

\end{proof}

 \vspace{7 mm}

\begin{thm}
\label{generalized_uncertainty_theorem3}
 If $w_1(x)=w_2(x)$ and the weights $w_1(x)$ and $w_2(x)$ are homogeneous of degree $k\in N$
 then

\begin{equation}
\psi_{w_1(x),w_2(\omega)}^m(\gamma^m,\zeta^m)=C(\gamma^m\zeta^m)^{\frac{-2}{km}}
\end{equation}
The sets $PB^k_{\gamma^m,\zeta^m}$ $m\in N$ are convex. Either all
of them are open or all of them are closed
\end{thm}

%\prf
\begin{proof}

First we show that $\forall c_0\in R^+$ a point $(\gamma,1,c_0)\in
PB_{\gamma_1,\zeta_1}$ iff the point
$(\frac{1}{\zeta_0}\gamma,\zeta_0,c_0)\in PB_{\gamma_1,\zeta_1}$:

 \vspace{3 mm}

$(\gamma,1,c_0)\in PB_{\gamma_1,zeta_1}$ implies that $\exists
f\in L^2$ such that for $a=c_0$, $b=1$ we have:

\begin{equation}
\gamma(f,c_0)=\frac{(\int_{-\infty}^\infty
{w_1}_{c_0}(x)|f(x)|^2dx)^\frac{1}{2}}{||f||}=\frac{c_0^\frac{-k}{2}(\int_{-\infty}^\infty
{w_1}_{1}(x)|f(x)|^2dx)^\frac{1}{2}}{||f||}=\gamma
\end{equation}

and

\begin{equation}
\zeta(f,1)=\frac{(\int_{-\infty}^\infty
{w_2}_{1}(\omega)|\hat{f}(\omega)|^2d\omega)^\frac{1}{2}}{||f||}=1
\end{equation}

then for the same function $f\in L^2$ and
$a=c_0\zeta_0^{\frac{2}{k}}$, $b=\zeta_0^{\frac{2}{k}}$ we have:

\begin{equation}
\gamma(f,c_0 \zeta_0^{\frac{2}{k}})=\frac{(\int_{-\infty}^\infty
\frac{1}{(c_0 \zeta_0^{\frac{2}{k}})^k}
{w_1}_{1}(x)|f(x)|^2dx)^\frac{1}{2}}{||f||}=\frac{1}{\zeta_0}\gamma(f,c_0)=\frac{1}{\zeta_0}\gamma
\end{equation}

and

\begin{equation}
\zeta(f,\zeta^{\frac{-2}{k}})=\frac{(\int_{-\infty}^\infty
{w_2}_{\zeta^{\frac{-2}{k}}}(\omega)|\hat{f}(\omega)|^2d\omega)^\frac{1}{2}}{||f||}=\zeta_0
\zeta(f,1)=\zeta_0
\end{equation}

which implies that $(\frac{1}{\zeta_0}\gamma, \zeta_0,c_0)\in
PB_{\gamma^1,\zeta^1}$.

The other direction is done in a similar way. \vspace{3 mm}

From the definition of
$\psi_{w_1(x),w_2(\omega)}^{\gamma^1}(\zeta^1,c)$
\eqref{psi_with_different_indices} we get that

\begin{equation*}
\psi_{w_1(x),w_2(\omega)}^{\gamma^1}(\zeta_0,c_0)=\frac{1}{\zeta_0}\psi_{w_1(x),w_2(\omega)}^{\gamma^1}(1,c_0)
\end{equation*}

and that if the infimum is attained in one point, say
$(\zeta^*,c_0)$, then it is attained for every $(\zeta,c_0)$,
$\zeta \in R^+$.

From lemma \ref{homogeneity_lemma} we get

\begin{equation}
\psi_{w_1(x),w_2(\omega)}^{\gamma^1}(\zeta,c)=c^\frac{-k}{2}\psi_{w_1(x),w_2(\omega)}^{\gamma^1}(\zeta,1)=\psi_{w_1(x),w_2(\omega)}^{\gamma^1}(1,1)\frac{1}{\zeta}c^\frac{-k}{2}
\end{equation}

and that if the infimum is attained for all the points
$(\zeta,c_0)$, $\zeta\in R^+$, then it is attained for all the
points $(\zeta,1)$, $\zeta\in R^+$, and therefore (using Lemma
\ref{homogeneity_lemma} again) it is attained for every point
$(\zeta,c)\in R^{2+}$.

From the explicit formula for
$\psi_{w_1(x),w_2(\omega)}^{\gamma^1}(\zeta,c)$ we see that the
graph of $\psi_{w_1(x),w_2(\omega)}^{\gamma^1}(\zeta,c)$ is
concave (as a multiplication of two concave functions; the proof
of this simple fact is similar to the case of a sum of two concave
functions that appears as part of the proof of Theorem
\ref{LPuncertainty}) and that it is equal to $\partial
PB_{\gamma^1,\zeta^1}$. From what we have showed above we get that
$PB_{\gamma^1,\zeta^1}$ is closed iff $\exists v\in R^3$ s.t.
$v\in
\partial PB_{\gamma^1,\zeta^1}\cap PB_{\gamma^1,\zeta^1}$.

Thus we get that $PB_{\gamma^1,\zeta^1}$ is convex and either open
or closed.

From convexity of $PB_{\gamma^1,\zeta^1}$ and the fact that
$\psi_{w_1(x),w_2(\omega)}^{\gamma^1}(\zeta^1,c)$,

$\psi_{w_1(x),w_2(\omega)}^{\zeta^1}(\gamma^1,c)$ and
$\psi_{w_1(x),w_2(\omega)}^1(\gamma^1,\zeta^1)$ are defined on
$(0, \infty) \times (0, \infty)$ (see Lemma \ref{argument_1} and
Definition \ref{psi_with_different_indices}) we get that their
graphs coincide and equal to:

\begin{equation}
\partial PB_{\gamma^1,\zeta^1}=\big\{(\gamma,\zeta,c)\in R^{3+}|
\gamma=\psi_{w_1(x),w_2(\omega)}^{\gamma^1}(1,1)\frac{1}{\zeta}c^\frac{-k}{2}\big\}
\end{equation}

and we get

\begin{equation*}
\psi_{w_1(x),w_2(\omega)}^1(\gamma^1,\zeta^1)=C(\gamma
\zeta)^\frac{-2}{k}
\end{equation*}

substituting $\gamma^1=1$ and $\zeta^1=1$ we find that
$C=\psi_{w_1(x),w_2(\omega)}^1(1,1)$

Thus we have

\begin{equation*}
\psi_{w_1(x),w_2(\omega)}^1(\gamma^1,\zeta^1)=\psi_{w_1(x),w_2(\omega)}^1(1,1)(\gamma
\zeta)^\frac{-2}{k}
\end{equation*}

and

\begin{equation}
\label{psi_m2}
\psi_{w_1(x),w_2(\omega)}^m(\gamma^m,\zeta^m)=\psi_{w_1(x),w_2(\omega)}^1(1,1)(\gamma^m\zeta^m)^{\frac{-2}{km}}
\end{equation}

From the relation between the possibility bodies
\eqref{possibility_bodies_relation} and from \eqref{psi_m2} it is
easy to see that: $\forall k\in N$ the sets,
$PB^k_{\gamma^m,\zeta^m}$ $m\in N$ are convex; they are open iff
the set $PB^k_{\gamma^1,\zeta^1}$ is open and closed iff the set
$PB^k_{\gamma^1,\zeta^1}$ is closed.
\end{proof}

When we take $k=2$ and $m=1$ in Theorem
\ref{generalized_uncertainty_theorem3} we get the
Heisenberg-Pauli-Weyl general uncertainty principle

\begin{equation*}
\psi_{w_1(x),w_2(\omega)}^1(\gamma^1,\zeta^1)=C(\gamma \zeta)^{-1}
\end{equation*}

In Subsection \ref{subsec_Heisenberg-Pauli-Weyl} we also managed
to compute the constant $C=\frac{1}{4\pi}$ for this special case
(see \eqref{HPW_theta_function} and \eqref{HPW_psi_fumction}).

\end{document}